\documentclass[USenglish,a4paper,11pt]{article}
\pdfoutput=1

\input{style.sty}
\title{The NISQ Complexity of Collision Finding\footnote{A previous version of this paper, which only covered model~1, was titled ``Quantum-Classical Tradeoffs in the Random Oracle Model''~\cite{HLS22p}.}}

\author{Yassine Hamoudi\thanks{Université de Bordeaux, CNRS, LaBRI, \texttt{ys.hamoudi@gmail.com}.} \qquad Qipeng Liu\thanks{University of California at San Diego, \texttt{qipengliu0@gmail.com}.} \qquad Makrand Sinha\thanks{University of Illinois at Urbana-Champaign, \texttt{msinha@illinois.edu}.}\\ [12pt]}

\date{}

\begin{document}

\maketitle

\begin{abstract}
  Collision-resistant hashing, a fundamental primitive in modern cryptography, ensures that there is no efficient way to find distinct inputs that produce the same hash value. This property underpins the security of various cryptographic applications, making it crucial to understand its complexity. The complexity of this problem is well-understood in the classical setting and $\Theta(N^{1/2})$ queries are needed to find a collision. However, the advent of quantum computing has introduced new challenges since quantum adversaries --- equipped with the power of quantum queries --- can find collisions much more efficiently. Brassard, Høyer and Tapp \cite{BHT98c} and Aaronson and Shi \cite{AS04j} established that full-scale quantum adversaries require $\Theta(N^{1/3})$ queries to find a collision, prompting a need for longer hash outputs, which impacts efficiency in terms of the key lengths needed for security.

This paper explores the implications of quantum attacks in the Noisy-Intermediate Scale Quantum (NISQ) era. In this work, we investigate three different models for NISQ algorithms and achieve \emph{tight bounds for all of them}:
\begin{enumerate}
    \item A hybrid algorithm making adaptive quantum or classical queries but with a limited quantum query budget, or
    \item A quantum algorithm with access to a noisy oracle, subject to a dephasing or depolarizing channel, or
    \item A hybrid algorithm with an upper bound on its maximum quantum depth; \emph{i.e.} a classical algorithm aided by low-depth quantum circuits.
\end{enumerate}
In fact, our results handle all regimes between NISQ and full-scale quantum computers. Previously, only results for the preimage search problem were known for these models (by Sun and Zheng~\cite{SZ19p}, Rosmanis~\cite{Ros22p,Ros23p}, Chen, Cotler, Huang and Li~\cite{CCHL23j}) while nothing was known about the collision finding problem.

Along with our main results, we develop an information-theoretic framework for recording query transcripts of quantum-classical algorithms. The main feature of this framework is that it allows us to record queries in two incompatible bases --- classical queries in the standard basis and quantum queries in the Fourier basis --- consistently. We call the framework the \emph{hybrid compressed oracle} as it naturally interpolates between the classical way of recording queries and the {compressed oracle} framework of Zhandry for recording quantum queries. We demonstrate its applicability by giving simpler proofs of the optimal lower bounds for NISQ preimage search and by showing optimal lower bounds for NISQ collision finding.

\end{abstract}
\newpage

\section{Introduction}
In modern cryptography, collision-resistant hashing stands as a cornerstone, providing countless cryptographic protocols and systems. Collision resistance refers to the intractability of recovering two distinct inputs that produce the same hash value. Collision resistance is crucial to establishing the security of many cryptographic applications, including digital signatures~\cite{KL07b}, Merkle trees~\cite{Mer89c}, zero-knowledge proofs/arguments~\cite{BG09j}, and many more. Thus, understanding collision resistance (or the complexity of collision finding) is particularly important to understand the security of these cryptographic applications.

The so-called generic attacks or black-box query model has received a lot of attention in understanding the security of various cryptographic primitives. In this model, when working with hash functions, an algorithm can only take advantage of the input-output behavior of the function, but does not have access to its actual implementation or other side-information. This approach not only provides simpler proof techniques, but indeed effectively encapsulates real-world attack scenarios. Classically, the complexity of collision finding in the black-box model is well understood. For instance, when employing an ideal hash function, denoted as $F:[M] \to [N]$, the best possible attack needs to make $\Omega(\sqrt{N})$ queries to the hash function to find a collision pair with high probability\footnote{We remark that for typical applications the parameter $M$ satisfies $M=\Omega(N)$.}, aligning with the upper bound implied by the Birthday problem.
In practical applications, it is imperative that adversaries with limited resources, typically no more than $2^{128}$ units of computational time, are unable to find a collision. This requirement necessitates a minimum output length of at least $256$ bits. As an illustrative example, the latest addition to the Secure Hash Algorithm family of standards, SHA3-256, as released by NIST, frequently finds use in such applications.

The emergence of quantum computing requires us to significantly reevaluate existing cryptography since quantum adversaries can be much more powerful. In the quantum black-box model, \emph{quantum queries}, i.e. the ability to access in superposition the values of a black-box function~\cite{BdW02j} is treated as a fundamental resource. This idealized input model gave rise to the early quantum algorithms by Deutsch and Jozsa~\cite{DJ92j}, Simon~\cite{Sim97j} (paving the way for Shor's factoring algorithm~\cite{Sho97j}), and Bernstein and Vazirani~\cite{BV97j}.

For collision finding, how should we set the output length such that the hash function is still collision resistant even against quantum adversaries? Brassard, Høyer and Tapp~\cite{BHT98c} and Aaronson and Shi~\cite{AS04j} proved that in the quantum black-box model, $\Theta(N^{1/3})$ queries are both sufficient and necessary for finding a collision. This suggests that to maintain the same level of security (i.e., secure against quantum algorithms that run in time $2^{128}$), the output length of the hash function needs to be extended to $3 \times 128 = 384$. Consequently, this adjustment in output length has affected storage requirements and the overall efficiency of various cryptographic protocols. However, as we are in the noisy-intermediate scale quantum (NISQ) era, quantum computation is noisy and quantum memory is short-lived, we ask the following question:
\begin{center}
\it{Should we sacrifice efficiency for potential quantum attacks, especially in the NISQ era?}
\end{center}
The question above has a natural motivation stemming from practice. In particular, various constraints on near-term quantum hardware often necessitate the use of classical processing in addition to quantum operations. For instance, in certain scenarios, it might not be feasible that superposition queries could be made to the entire input, or the cost of making such queries might be prohibitive. Furthermore, the depth of possible quantum computation in near-term devices is also limited since the decoherence effects accumulate, thus additional classical processing is warranted to fully utilize the capabilities of such devices.

Motivated by the above considerations, in this paper, we investigate the limitations of NISQ algorithms for collision finding, as well as introduce a general technique/framework for proving lower bounds in the NISQ era. Using the new framework, along with the lower bounds for collision finding, we also give simpler and unified proofs of several results on preimage finding by Sun and Zheng~\cite{SZ19p}, Rosmanis~\cite{Ros22p,Ros23p}, and Chen, Cotler, Huang and Li~\cite{CCHL23j}.

%%%%%%%%%%%%%%%%%%%%%%%%%%%%%%%%%%%%%%%%%%%%%%%%%%%%%%%%%%%%%%%%%%%%%%%%%%%%%%%%

\subsection{Contributions}

We first present our contributions on the limitations of NISQ collision finding.

\paragraph*{Collision Finding (Section~\ref{sec:hybrid_cols}).}

The problem is to find a pair of elements $x \neq y \in [M]$ that evaluate to the same value $F(x) = F(y)$, given a uniformly random function $F : [M] \ra [N]$. Classically, a tight bound $\Theta(c^2/N)$ for the optimal success probability is easily proved for this problem, where $c$ is the number of classical queries. When a full-scale quantum computer with $q$ quantum queries is available, one can use the so-called BHT algorithm~\cite{BHT98c} to find a collision pair with probability $\bo{q^3/N}$ (assuming $M = \om{N^{2/3}}$). However, this algorithm requires $q$ quantum queries, meaning no noise or upper bounds on quantum depth were considered for implementing the algorithm, leaving the potential quantum speed-up elusive in the NISQ era. Towards resolving this issue, we propose three different models for NISQ algorithms and show tight bounds for each of these models.

\paragraph{Model 1. Bounded Quantum Queries.}
In this model, we consider a quantum algorithm that only has limited access to its quantum capabilities: namely, an upper bound on the number of quantum queries, denoted by $q$. Additionally, the algorithm can make (potentially significantly more)~$c$ classical queries. This model is closely related to ``$d$-QC model'' discussed in the line of work~\cite{CCL23j,CM20c,AGS22p,HG22c}, where $d$ quantum queries are interleaved with
classical queries. Rosmanis~\cite{Ros22p} proved a tight bound for preimage search in this model, and posed an open question on collision finding. We answer this question below:

\begin{rtheorem}[Theorem~\ref{Thm:qc-coll}, first bullet]
  The optimal success probability of an algorithm making $q$ quantum and $c$ classical queries for solving the Collision Finding problem is $\ta{(c^2 + cq^2 + q^3)/N}$. There is a matching hybrid algorithm that achieves asymptotically the same success probability.
\end{rtheorem}
\noindent Our bound is tight when $M = \om{N^{2/3}}$ because of the following variant of the BHT algorithm: the first $c+\lceil{q/2}\rceil$ queries are classical\footnote{The first $\lceil{q/2}\rceil$ quantum queries are also used to make classical queries.} and are used to collect distinct $(x, F(x))$ pairs. If there is a collision among these values the algorithm terminates. Otherwise, the remaining $\lfloor{q/2}\rfloor$ quantum queries are used to run Grover's search on the rest of $F$, where an element $x$ is marked if its image $F(x)$ occurs among the collected pairs. This algorithm stops working for small domains of size~$M = \bo{N^{2/3}}$, as is the case for the BHT algorithm. In fact, we conjecture that the optimal bound is $\ta{(c^2+q^3)/N}$ when~$M = \ta{\sqrt{N}}$ (which is the regime of the Element Distinctness problem). We also note that $M > N$ is safe to assume for most cryptographic applications.

\paragraph{Model 2. Noisy Quantum Queries.}
In the second model, we consider noisy quantum machines, whose only noise comes from quantum queries to the hash function. More explicitly, we assume each quantum query to the hash function is affected by a dephasing noise $b \in (0, 1]$: with probability $1-b$, it is a quantum query; otherwise (with probability $b$), it is a classical query. We ignore all other noises in this model and only pose constraints on oracle queries.

\begin{rtheorem}[Theorem~\ref{Thm:qc-coll}, second bullet]
   The optimal success probability of an algorithm making $t$ noisy queries with dephasing noise $b \in [\frac1t, 1]$ for solving the Collision Finding problem is~$\ta{t^2/(b N)}$. There is a matching hybrid algorithm that achieves asymptotically the same success probability.
\end{rtheorem}
\noindent In the above theorem statement, we only consider $b \geq \frac1t$, as when $b$ is sufficiently small,~$t$ noisy queries are already very likely to be all quantum.
Our bound is tight when $M = \om{N^{2/3}}$ due to the following variant of the BHT algorithm. First, make~$t/2$ classical queries to collect distinct $(x, F(x))$ pairs (note that a noisy query can be made purely classical by simply measuring both the input and output registers). Next, run~$bt$ independent instances of Grover's search, each using $1/(2b)$ noisy queries, and try to find a collision within the collected pairs. As there are only $1/(2b)$ noisy queries, each instance is a purely quantum algorithm with high probability. Thus, each instance succeeds with probability $\om{1/(2b)^2 \cdot t/(2N)} = \om{t/(b^2N)}$. Consequently, this algorithm succeeds with probability $\om{bt \cdot t/(b^2N)} = \om{t^2/(b N)}$.

\paragraph{Model 3. Bounded Quantum Depth.} In this model, we consider a quantum algorithm that is almost classical, but has access to quantum helper subroutines that have bounded depth~$d$. In other words, the collection of algorithms in this model can be modeled as ${\sf BPP}^{{\sf QNC}_d}$. This model captures a significant NISQ scenario where we have access to arbitrary polynomial-time randomized classical algorithms, but all usable quantum machines are vulnerable to noise and completely collapse after a certain period of time. This model is the ``$d$-CQ model'' discussed in the line of work~\cite{CCL23j,CM20c,AGS22p,HG22c}.

\begin{rtheorem}[Theorem~\ref{Thm:qc-coll}, third bullet]
   The optimal success probability of an algorithm making~$t$ quantum queries with bounded depth $d \leq t$ for solving the Collision Finding problem is $\ta{dt^2/N}$. There is a matching algorithm that achieves asymptotically the same success probability.
\end{rtheorem}
\noindent
The tightness of the bound follows from a similar algorithm to the one in model 2 when~$b = 1/d$.

\vspace*{5pt}

Our results are proven using a new information-theoretic framework that we call the hybrid compressed oracle. We next provide a high-level description of this framework.

\paragraph*{Hybrid Compressed Oracle (Section~\ref{sec:compressed_oracle_for_hybrid}).}

The main technical contribution of this work is a new information-theoretic lower-bound framework, called the \emph{hybrid compressed oracle}, for analyzing the success probability of hybrid algorithms that perform a mix of quantum and classical queries. As the name suggests, our framework is an extension of the compressed oracle technique of Zhandry~\cite{Zha19c}. This part of our work broadly fits under the long-term goal in complexity theory to develop general lower-bound techniques that characterize the \emph{tradeoffs} between the number of queries and other computational resources. For instance, prior works have studied the interplay between quantum queries and memory space~\cite{KSW07j,ASW09j,HM23j}, circuit depth~\cite{SZ19p,CCL23j,CM20c,AGS22p,HG22c,CH22pa}, parallel computation~\cite{Zal99j,GR04p,JMdW17j,AHU19c,CFHL21c,BLZ21c}, proof size~\cite{Aar12j,ST23j,AKKT20c}, advice size~\cite{NABT15j,HXY19c,CLQ20c,CGLQ20c,GLLZ21c}, among others. These results are often tailored to the problems at hand and do not provide general lower-bound frameworks however.

Our hybrid lower-bound framework departs from a recent method introduced by Zhandry \cite{Zha19c}, called the \emph{compressed oracle} (see Section~\ref{sec:overview-compressed}), that quantizes the classical \emph{lazy sampling} technique. The classical variant of the method records a \emph{query transcript} representing the knowledge gained by an algorithm (the ``attacker'') on the input and on an intuitive level uses it to argue that the algorithm does not record enough information via these queries to succeed. However, the recording of quantum queries is a blurry task to define due to the no-cloning theorem and the superposition input access. Some important features of Zhandry's solution to these problems are the construction of a quantum query transcript in the Fourier domain, and the ability for the attacker to erase the transcript (for instance, by running its algorithm in reverse).

We first extend Zhandry's construction to support recording both classical and quantum queries. This is not as easy as it may seem since it requires merging two ways of recording on distinct bases (the standard and the Fourier basis). Our solution relies on replacing the original classical and quantum query operators with two ``recording query operators'' (Section~\ref{Sec:construction}) that maintain a consistent classical-quantum query transcript throughout the execution of the algorithm (Proposition~\ref{Prop:consistency}). In the extreme cases where all the queries are classical or quantum, our framework recovers the classical lazy sampling and the quantum compressed oracle techniques, respectively. Moreover, as in previous work, our hybrid recording perfectly simulates the behavior of the original algorithm (Proposition~\ref{Prop:uncomp}).

We then further extend our framework to record \emph{mixtures} of the classical and quantum oracles. Such mixtures capture the model where a quantum query can collapse into a classical one because of dephasing noise.
We handle this setting by interpolating between the two types of recording that happen when the query is purely classical or quantum. Our simulation is again indistinguishable from the viewpoint of the algorithm.
Furthermore, we demonstrate a close connection between a mixture that puts probability $b \in (0,1]$ on the classical oracle and the model where the quantum depth is bounded by $1/b$.
The latter amounts to a complete collapse of the quantum memory after every $1/b$ quantum queries.
We show that, when replacing each quantum query with the aforementioned mixture and removing the depth constraint, the success probability of the algorithm is barely changed. Hence, the depth-bounded model can be analyzed in our framework using the appropriate interpolation parameter.
A more detailed technical overview of our framework is provided in Section~\ref{sec:overview-hybrid}.\\

Apart from proving NISQ lower bounds for collision finding, we also demonstrate the applicability of our framework by proving NISQ lower bounds for preimage search in all three models. These lower bounds were previously shown by \cite{SZ19p, CCHL23j,Ros22p,Ros23p} and we are able to give unified and simplified proofs of these results.

\paragraph*{Preimage Search (Section~\ref{sec:hybrid_search}).}

The preimage search concerns the problem of finding a preimage $x \in [M]$  satisfying $F(x) = 0$ given a uniformly random function $F : [M] \ra [N]$. The optimal success probability for solving this problem is~$\ta{c/N}$ with $c$ classical queries, or~$\ta{q^2/N}$ with $q$ quantum queries by using Grover's algorithm~\cite{Gro97j}. Rosmanis \cite{Ros22p},  using a proof tailored to the search problem, showed that no hybrid algorithm can interpolate between these two cases efficiently. Here, we give a simpler proof of the same result using the hybrid compressed oracle framework.

\begin{rtheorem}[Theorem~\ref{Thm:qc-search}, first bullet]
  The optimal success probability of an algorithm making~$q$ quantum and $c$ classical queries for solving the Preimage Search problem is $\ta{(c + q^2)/N}$.
\end{rtheorem}

The proof relies on a simple application of our hybrid compressed oracle framework, where the progress made towards finding a preimage is represented as the probability of measuring a classical-quantum query transcript containing such a preimage. The central argument in our analysis, that allows us to overcome the $\bo{(c+q)^2/N}$ upper bound derived from the original compressed oracle, is a refinement of certain triangle inequalities when a classical query is made.

Sun and Zheng~\cite{SZ19p}, Chen, Cotler, Huang and Li~\cite{CCHL23j} and Rosmanis~\cite{Ros23p} also considered the case of hybrid algorithms that make noisy queries or have bounded depth. We also recover these results using the hybrid compressed oracle framework.

\begin{rtheorem}[Theorem~\ref{Thm:qc-search}, second bullet]
   The optimal success probability of an algorithm making $t$ noisy queries with dephasing noise $b \in [1/t, 1]$ for solving the Preimage Search problem is $\ta{t/(b N)}$. There is a matching algorithm that achieves asymptotically the same success probability.
\end{rtheorem}

\begin{rtheorem}[Theorem~\ref{Thm:qc-search}, third bullet]
   The optimal success probability of an algorithm making~$t$ quantum queries with bounded depth $d \leq t$ for solving the Preimage Search problem is $\ta{dt/N}$. There is a matching hybrid algorithm that achieves asymptotically the same success probability.
\end{rtheorem}

%%%%%%%%%%%%%%%%%%%%%%%%%%%%%%%%%%%%%%%%%%%%%%%%%%%%%%%%%%%%%%%%%%%%%%%%%%%%%%%%

\subsection{Related Work}

\paragraph{Query Complexity Lower Bounds.} There are two main systematic techniques for proving lower bounds in quantum query complexity: the polynomial~\cite{BBC+01j} and the adversary~\cite{Amb02j} methods. A different information-theoretic method, called the compressed oracle technique, was recently introduced by Zhandry~\cite{Zha19c}.  This method is useful in proving lower bounds for search problems when the superposition queries are made to a \emph{uniformly random} function, a setting that is often used to model various cryptographic scenarios, and is commonly called \emph{the random oracle model} in the cryptography literature. The compressed oracle technique has led to new and simpler lower bounds for certain search problems (e.g.~\cite{LZ19c,HM23j,Ros21p}) and security proofs in post-quantum cryptography (e.g.~\cite{HI19c,LZ19ca,CMS19c,CMSZ19p,BHH+19c,AMRS20c}).

While the classical counterparts of these methods are often easy to manipulate, it is generally unknown how to adapt them to the hybrid setting. Indeed, the only prior works concerning the hybrid quantum-classical query model, that we are aware of, use ad-hoc methods that are tailor-made for the specific problem being studied.

\paragraph{Lower Bounds for Hybrid and NISQ Algorithms.}
As mentioned before, in a recent work, Rosmanis~\cite{Ros22p} characterized the optimal success probability of solving the preimage search problem, although not in the random oracle model. The proof techniques are specific to the search problem and inspired by a lower bound for Grover's search with quantum faulty oracles by Regev and Schiff~\cite{RS08c}.

Another related line of work ~\cite{CCL23j,CM20c,AGS22p,HG22c,CH22pa} proves lower bounds for hybrid algorithms in the so-called ``$d$-QC model''  where $d$ quantum queries are interleaved with polynomially (in the number of input qubits) many classical queries. This model is akin to small-depth \emph{measurement-based quantum computation}, where measurement outcomes are classically processed to select subsequent quantum gates and is encompassed by our hybrid quantum-classical query model when the number of quantum queries is bounded by $d$. This model is captured by the first hybrid model with a bound on the number of quantum queries. The aforementioned works show that certain carefully constructed variants of Glued-Trees and Simons problems require a large quantum depth. For the preimage search problem, Sun and Zheng~\cite{SZ19p}, Chen, Cotler, Huang and Li~\cite{CCHL23j} and Rosmanis~\cite{Ros23p} also considered the case of hybrid algorithms that make noisy queries or have bounded depth, as mentioned above.

In post-quantum cryptography, several works~\cite{JST21c,ABKM22c} studied the post-quantum security of the Even-Mansour and FX constructions when the attacker has quantum access to the underlying block cipher and classical access to the keyed primitive. These results are based on new ``reprogramming'' lemmas for analyzing the advantage of distinguishing between two oracles that differ in some specific way. Additionally,~\cite{JST21c} introduced a variant of the compressed oracle for recording both classical and quantum queries \emph{in the Fourier domain}. It allows the authors to argue that, for a variant of the FX construction, the classical and quantum queries can be (approximately) treated as acting on disjoint domains. This method does not seem generalizable to the proof of more general hybrid results.

\section{Technical Overview}
\label{sec:overview}

\subsection{Overview of the Compressed Oracle}
\label{sec:overview-compressed}

First we give a detailed overview of the compressed oracle framework \cite{Zha19c}. As mentioned before, this framework gives an information-theoretic method that is useful in proving lower bounds against quantum algorithms that get black-box query access to a uniformly random function $F : [M] \ra [N]$. The framework allows one to store a compressed encoding of the uniformly random function conditioned on the knowledge gained from the queries.

To illustrate the framework, we first consider the case of classical and quantum algorithms separately and then discuss the ideas involved in extending the framework to the setting of hybrid algorithms. For pedagogical reasons, we shall primarily focus on the preimage search problem as a running example and use $D$ (instead of $F$) to denote a uniformly random function (or database) henceforth.

\paragraph{Classical Algorithms.} Let us first consider classical query algorithms for the search problem. After $c$ classical queries at most $c$ entries of the uniformly random function $D$ can be assumed to be fixed, since the entries that have not been queried are still uniformly random in $[N]$. This observation allows one to model the random function $D$ as being generated by \emph{lazy sampling}: we may think of a location $x \in [M]$ that has not been queried to be marked with a special symbol~$\bot$ and whenever that location is queried for the first time, $D(x)$ is replaced with a uniformly random value in $[N]$. In other words, after $c$ queries, we store a compressed encoding of $D$ where only $c$ locations are fixed, and others are compressed to a special symbol $\bot$. Whenever a query is made to a location that is still compressed, it is uncompressed and replaced by a uniformly random value. It follows that if after $c$ queries we have not seen a zero preimage, then the probability of seeing a zero preimage in the next query is $1/N$. Thus the probability of success after $t$ queries, denoted $p_t$, satisfies $p_{t+1} \le p_t + 1/N$ and is bounded by $c/N$ after $c$ queries.

\paragraph{Quantum Algorithms.} The compressed oracle framework quantizes the lazy sampling idea and allows one to define a compressed encoding of a random function that works well with quantum queries. Unlike the classical case, quantum information can not be cloned and could be forgotten, so some care needs to be taken in defining this compressed encoding. Consider a quantum algorithm that has an index register $\CX$, a phase register $\CP$, a workspace register $\CW$ and has black-box access to a uniformly random function $D$ via the following phase\footnote{Note that the value of $D(x)$ is returned in the phase of the complex state and $p$ is an additional control register. This kind of query is usually called a \emph{phase query} in the literature. There is another standard way of defining a quantum query by a unitary that maps $\ket{x,p,w}$ to $\ket{x,p \oplus D(x),w}$. The two kinds of queries are equivalent up to a unitary transformation, and we focus on phase queries as they work better with our framework.} unitary:
\[\mc{O}^Q_D : \ket{x,p,w} \mapsto \omega_N^{p D(x)} \ket{x,p,w} \quad \mbox{where} \quad \omega_N = e^{\frac{2{\bf i}\pi}{N}}.\]

A quantum algorithm starts with the all-zero state $\ket{0,0,0}$ and applies arbitrary unitaries interleaved with phase queries. For a fixed $D: [M] \to [N]$, the state of the algorithm at any point is some arbitrary state $\ket{\psi_D}$. After averaging over uniformly random $D$, the state is the mixed state $\mathbb{E}_D[\ket{\psi_D}\bra{\psi_D}]$ and  it will be more convenient for us to work with a purification of this state. We add a purification register $\CD = \CD_0 \cdots \CD_{M-1}$  where the subregister $\CD_x$ for $x \in [M]$ holds a value $D(x) \in [N]$ and we refer to it as the database register. Then, the  state
  \[\frac1{N^{M/2}} \sum_{D \in [N]^M} \ket{\psi_D}\otimes \ket{D}_{\CD}\]
is a purification, as after tracing out $\CD$ we obtain the same mixed state as before. Note that in the above encoding, the database register is never altered during the run of the algorithm.

Motivated by the classical case, we would like to have a compressed encoding of the random function $D$. For this, we extend the range of $D$ to allow for a compressed symbol $\bot$ and define compression and uncompression operations that act on the database register $\CD$ whenever a query is made. In particular, let $D: [M] \to \{\bot\} \cup [N]$ and extend the register $\CD_x$ so that it can now also hold the value $\bot$. The initial state of the register $\CD$ (at the beginning of the algorithm) is  $\ket{\bot,\ldots,\bot}_{\CD}$, which corresponds to a completely compressed database. We also define a Hermitian unitary operation $S$ that is controlled on the index register $\CX$ and uncompresses an entry that is $\bot$: if the index register is $\ket{x}_\CX$ and the database register is $\ket{\bot}_{\CD_x}$, then it is mapped to $\frac1{\sqrt{N}}\sum_{y \in [N]} \ket{y}_{\CD_x}$ while it maps the last state back to $\ket{\bot}_{\CD_x}$ (for details on how to unitarily implement this, see \Cref{sec:compressed_oracle_for_hybrid}). Before a quantum query, the database is uncompressed by applying $S$ and after the query it is compressed again by applying $S$.

With the above framework, one can prove a lower bound for the preimage search problem against any quantum algorithm by following a similar template as in the classical case. In particular, the probability $p_t$ of succeeding after $t$ queries is essentially the \emph{squared} norm of the projection of the state at time $t$ onto the subspace spanned by databases $\ket{D}_{\CD}$ that contain a zero preimage. One can show that the norm of this projection is initially $0$ and increases by at most $\bo{1/\sqrt{N}}$ after each query and thus
  \[\sqrt{p_{t+1}} \le \sqrt{p_t} + \bo*{\frac1{\sqrt{N}}} \implies p_q = \bo*{\frac{q^2}N}.\]

We stress out that the compressed encoding is just a technique for proving lower bounds in the real random oracle model. An algorithm will never encounter the compressed symbol $\bot$ in practice, as the simulation is statistically indistinguishable from the real world.

%%%%%%%%%%%%%%%%%%%%%%%%%%%%%%%%%%%%%%%%%%%%%%%%%%%%%%%%%%%%%%%%%%%%%%%%%%%%%%%%

\subsection{Overview of the Hybrid Compressed Oracle}
\label{sec:overview-hybrid}

One of the main contributions of this work is to extend the compressed oracle framework to the setting of hybrid algorithms that make both quantum and classical queries. In fact, we consider an even more general scenario where each query can behave as a superposition of the quantum and classical oracles according to some interpolation parameter. This setting allows us to capture a wide variety of NISQ models based on noisy oracles and depth-bounded quantum algorithms.

Since a quantum query can always simulate a classical query, one could hope to analyze such algorithms using the compressed oracle framework for quantum algorithms above. However, it is not straightforward in such an analysis to capture that classical queries do not create additional interference. In fact, such attempts run into significant technical difficulties.

Here we start from first principles and define another purification compatible with both classical and quantum queries and that allows us to store a compressed encoding of the random function $D$ conditioned on the queries made by the algorithm. There are two main principles behind the new purification that take into account the classical nature of the queries:
  \begin{description}
    \item[Measurement] Classical queries can be measured, so we add an additional history register $\CH$ that records all the classical queries $(x,D(x))$. The contents of a recorded query in this register are never changed.
    \item[Consistency] We define compression and uncompression operations for the database $\CD$ conditioned on the history. In particular, under the standard compressed oracle framework~$\ket{y}_{\CD_x}$ can be changed during (un)compression if the index register contains~$\ket{x}_{\CX}$, which captures the fact that quantum algorithms could forget information. However, in the new purification, if $(x, D(x))$ is in the history, which happens if $x$ has been queried classically, then the register $\CD_x$ is never compressed or uncompressed again.
  \end{description}

\paragraph{Lower Bound for Preimage search with classical/quantum queries (Model 1).} With the above framework, we give an alternative lower-bound proof for the search problem against hybrid algorithms that use $c$ classical and $q$ quantum queries (when $c = 0$ or $q = 0$ we recover the usual quantum or classical bounds). As remarked before, this was first shown by Rosmanis~\cite{Ros22p} with a proof tailored for the search problem. Although there are some similarities between his approach and ours, the proof using the hybrid compressed oracle framework follows in a more principled way, is arguably simpler and works in the random oracle model.

To prove the lower bound, we again bound  the probability $p_t$ of succeeding after $t$ queries. To do this, we now keep track of whether there is a zero preimage in the classical history or in the quantum database: let $\ket{\phi}$ be the current joint state of all registers, we define $\pc$ as the projector on the span of the basis state where the classical history $\CH$ contains a zero preimage and  $\pqnc$ as the projector on those basis states where there is a zero preimage in the quantum database $\CD$ but none in the history. The norms $\|\pc\ket{\phi}\|$ and $\|\pqnc\ket{\phi}\|$ can be considered the classical and quantum progress respectively.

We show that after a quantum query, the quantum progress $\|\pqnc\ket{\phi}\|$ increases by $\bo{1/\sqrt{N}}$ as in the completely quantum case, while the classical progress $\|\pc\ket{\phi}\|$ does not change. However, under a classical query, the classical progress could increase by a much larger amount, but only at the cost of decreasing the quantum progress. As an example, consider a hybrid algorithm that creates a superposition over all preimages of zero by performing Grover's search, then measures its internal register to get a random preimage $x$ and finally makes a classical query on $x$. Clearly, before the only classical query, we have  $\|\pc\ket{\phi}\| = 0$ and $\|\pqnc\ket{\phi}\| \approx 1$ but right after the query,~$\|\pqnc\ket{\phi}\|$ becomes almost zero whereas $\|\pc\ket{\phi}\| \approx 1$.

This phenomenon does not appear when the algorithm is purely classical or quantum. Nonetheless, upon making a classical query, we show that the total progress defined as
  \[\Psi_t = \|\pc\ket{\phi}\|^2 + 3 \|\pqnc\ket{\phi}\|^2\]
increases by at most $O(1/N)$, behaving as in the classical case. Note that $\Psi_t$ upper bounds the total probability of having a preimage in either the database or the classical history.

More precisely, let $\kpp$ be the resulting quantum state after a classical query is made. Although $\norm{\pc \kpp}^2$ can be much larger than $\norm{\pc \kp}^2$, the state $\pc \kpp$ consists of three parts:
\begin{enumerate}
    \item $\ket{\phi_1}$: This corresponds to the basis states that already contained a zero preimage in their history register prior to the last classical query. The squared norm of this part can be bounded by $\norm{\pc \kp}^2$.
    \item $\ket{\phi_2}$: This corresponds to the basis states where there was no zero preimage either in the history or the database (prior to the classical query) and the classical query sampled a new zero preimage. The squared norm of this term is roughly at most $1/N$.
    \item $\ket{\phi_3}$: The last part consists of the basis states where there was at least one zero preimage in the  database but none in the history (prior to the classical query) and the classical query either sampled a new preimage or ``moved'' one from the quantum database to the classical history. We denote the squared norm of $\ket{\phi_3}$ by $\delta_{\predQ \to \predC}$ (denoting the amplitude that moved from $\pqnc$ to $\pc$). This exactly captures the scenario mentioned in the above example using Grover's search.
\end{enumerate}

On a high level, we show that $\pc \kpp = \ket{\phi_1} + \ket{\phi_2} + \ket{\phi_3}$ and $\ket{\phi_1}$ is also  orthogonal to $\ket{\phi_2}$ and~$\ket{\phi_3}$. Thus, we have that
  \begin{align*}
    \norm{\pc \kpp}^2 ~=~ \norm{\ket{\phi_1}}^2 + \norm{\ket{\phi_2} + \ket{\phi_3}}^2 ~\le~ \norm{\pc \kp}^2 + 2\norm{\ket{\phi_2}}^2 + 2\norm{\ket{\phi_3}}^2.
  \end{align*}
The increase $\norm{\pc \kpp}^2 - \norm{\pc \kp}^2$ is then  $O(\frac{1}{N}) + 2\delta_{\predQ \to \predC}$.
On the other hand, $\norm{\pqnc \kp}^2$ will decrease by at least $\delta_{\predQ \to \predC}$ due to a similar reason.
Thus, we conclude that after a classical query,~$\Psi_t$ increases by at most $O(1/N)$ (in fact, $O(1/N) - \delta_{\predQ \to \predC}$ but we do not need that refinement here). Combined with the fact that a quantum query increases~$\sqrt{\Psi_t}$ by $O(1/\sqrt{N})$, this shows that the success probability after $c$ classical and $q$ quantum queries is at most~$\Psi_{c+q} = \bo[\big]{\frac{c+q^2}{N}}$.

\paragraph{Lower Bound for Preimage search with interpolated queries (Model 2).} We adapt the above proof to the case where each query is a mixture of the classical and quantum oracles (instead of being purely classical or quantum). For simplicity, we assume that all queries have the same probability~$b \in (0,1]$ of being classical, which is equivalent to making quantum queries affected by dephasing noise $b$.

The success probability of Grover's search using $t$ such queries is~$\om{(1-b)^t t^2/N}$ since the probability that all queries are quantum is $(1-b)^t$. This is nearly optimal when the noise is sufficiently small $b \leq 1/t$, as noiseless algorithms succeed with probability $\bo{t^2/N}$ anyway. However, when~$b \geq 1/t$, a better algorithm consists of running $\lfloor bt \rfloor$ independent instances of Grover's search, each using $\lfloor 1/b \rfloor$ queries, to succeed with probability $\om{bt \cdot 1/(b^2N)} = \om{t/(bN)}$.

We show that the above algorithm is optimal by tracking the same progress measure $\Psi_t$ as before, but now making interpolated queries. One can immediately apply the analysis of the previous paragraph to show that the progress increases by at most $\bo{(1-b) \sqrt{\Psi_t/N} + 1/N}$ after each query. This is however not sufficient to conclude that~$\Psi_t = O(t/(bN))$.
The proof involves refining the analysis of how the quantum progress changes after a query. We consider the exact value $\delta_{\predQ \to \br{\predQ}} = \norm{\pqnc \kp}^2 - \norm{\pqnc \kpp}^2$ by which it decreases when making a classical query. Since it is at least the amount $\delta_{\predQ \to \predC}$ transferred to the classical progress, we obtain that~$\Psi_t$ increases by at most $O(1/N) - \delta_{\predQ \to \br{\predQ}}$ after a classical query.
On the other hand, we show that the quantum progress increases by at most $\bo{\sqrt{\delta_{\predQ \to \br{\predQ}}/N} + 1/N}$ when making a quantum query, which is sometimes smaller than the quantity $\bo{\norm{\pqnc \kp}/\sqrt{N} + 1/N}$ used to analyze the model 1. Overall, by interpolating between the two oracles, we conclude that~$\Psi_t$ increases by
  \[\bo[\Big]{(1-b) \sqrt{\delta_{\predQ \to \br{\predQ}}/N} - b \cdot \delta_{\predQ \to \br{\predQ}} + 1/N} = \bo{1/(bN)}\]
after each interpolated query, since the function $Z \mapsto (1-b) Z/\sqrt{N} - b Z^2 + 1/N$ is at most $\bo{1/(bN)}$. Hence the success probability after $t$ queries is at most~$\Psi_t = \bo[\big]{\frac{t}{bN}}$.

\paragraph{Lower Bound for Preimage search with bounded depth (Model 3).}
At first sight, the bounded-depth model is more subtle to analyze since it concerns all the memory of the algorithm (which has to decohere every $d$ queries), instead of only the query registers. We do not know if a variant of the hybrid compressed oracle can capture this property optimally.
Instead, we aim to relax that model to focus the analysis on the query registers.
A first attempt could be to only decohere the latter registers, which amounts imposing a classical query every~$d$ quantum queries.
This is however a very weak constraint, since an algorithm can swap the query registers with garbage qubits before and after making the classical queries to avoid the decoherence.
Our solution is to instead show that a depth-bounded algorithm can always be simulated by an algorithm -- in model 2 -- where each query is classical with probability $b = 1/d$.
Intuitively, this amounts to ``spread out'' the decoherence occurring every $d$ queries (in the bounded-depth model) into a smaller probability $1/d$ of decohering \emph{only} the query registers but at \emph{every} query. The details of the reduction are provided in \Cref{Prop:reduction}. Plugging the parameter $b = 1/d$ in the above bound established in model 2, we immediately obtain that the success probability after $t$ queries in the bounded-depth model is at most $\bo[\big]{\frac{dt}{N}}$. This is easily shown to be optimal.

\paragraph{Lower Bounds for Collision Finding.} The intuition behind the proof for the collision lower bounds is similar to that for the search problem. However, the details are quite involved because of one crucial difference. For the preimage search problem, the preimage is either in the history~$\CH$ or only in the quantum database~$\CD$, allowing us to define classical and quantum measures of progress. For the collision finding problem, there could also be \emph{hybrid collisions}, meaning a colliding pair $(x,x')$ where $x$ is in the history while $x'$ is only in the database~$\CD$. This makes the proof substantially more involved, as one also needs to keep track of other progress measures for such hybrid collisions.

We only sketch the lower bound in model 1, where $c$ queries are classical and $q$ queries are quantum. The lower bounds in the two other models build upon these ideas in a similar way to what is discussed above for preimage search.

The proof consists again of bounding the probability $p_t$ of finding a collision after~$t$ queries. To do this, we now keep track of whether there is a classical, hybrid, or quantum collision. We define various projectors onto the span of basis states containing such collisions and use these as measures of classical, hybrid, or quantum progress.
Similar to the case of preimage search, a quantum query can only increase all these measures of progress by a small amount, but a classical query might increase some of them by a large amount while decreasing others at the same time. We are able to show how much amplitude is transferred onto the subspace spanned by basis states containing classical, hybrid or quantum collisions after making a quantum or classical query.

To be more precise, we define three projectors: $\pc, \phnc, \pqnhnc$. The support of $\pc$ consists of the span of all basis states whose classical history contains a collision. We similarly define~$\phnc$ for hybrid collisions only (no classical collisions) and $\pqnhnc$ for quantum collisions only (no hybrid or classical collisions). Similar to our discussion on preimage search, a classical query can move a large amplitude from $\phnc$ to $\pc$, or from $\pqnhnc$ to $\phnc$. This is more complicated than the case of preimage search, as there is a hierarchy of three projectors, instead of two in the prior case --- let $\kp$ be the current state and let $\Dc, \Dhnc, \Dqnhnc$ be the increment in the squared norms $\norm{\pc\kp}^2, \norm{\phnc\kp}^2, \norm{\pqnhnc \kp}^2$ after a classical query is made.
By a refinement of certain triangle inequalities, we show that:
  \begin{align*}
    \Dc &\leq 2 \delta_{\predH \to \predC} + \bo[\Big]{\frac{t}{N}},\\
    \Dhnc &\leq - \delta_{\predH \to \predC} + 2 \delta_{\predQ\to\predH} + \bo[\Big]{\frac{t}{N}},\\
    \Dqnhnc &\leq - \delta_{\predQ\to\predH} + \bo[\bigg]{\sqrt{\frac{t \cdot \delta_{\predH \to \predC}}{N}}}.
  \end{align*}

Using these facts, we prove that the following potential
  \[ \Psi_t := \norm{\pc\kp}^2 + 3 \norm{\phnc\kp}^2 + 7 \norm{\pqnhnc \kp}^2,\]
which upper bounds the total progress, always increases as follows: $\sqrt{\Psi_t} \le \sqrt{\Psi_{t-1}} + \bo[\Big]{\sqrt{\frac{t}{N}}}$ if the $t$-th query is quantum and $\Psi_t \le \Psi_{t-1} + \bo[\big]{\frac{t}{N}}$ if the $t$-th query is classical. Overall, this shows that the success probability of finding a collision after $c$ classical and $q$ quantum queries is at most $\Psi_{c+q} = \bo*{\frac{c^2+cq^2+q^3}{N}}$.

\section{Hybrid Random Oracle Model}
\label{sec:model}

Below we define a computational model that captures hybrid algorithms that make both classical and quantum queries to a random function (which we also refer to as a random oracle for consistency with the compressed oracle framework). We also note that our model captures the QC model \cite{CCL23j}, a generalized model for measurement-based quantum computation, as a special case.\footnote{In the QC model, there are $2q$ rounds of computation where in the even numbered rounds, $c/q$ classical queries are made, and in the odd numbered round, one quantum query is made followed by a (possibly partial) measurement. The measurements can be deferred till the end using ancilla qubits.}

\paragraph{Memory.}
The memory of an algorithm accessing an oracle $D : [M] \ra [N]$ is made of three quantum registers defined as follows:
  \begin{itemize}
    \item Index register $\mc{X}$ holding $x \in [M]$.
    \item Phase register $\mc{P}$ holding $p \in [N]$.
    \item Workspace register $\mc{W}$ holding $w \in \rn^*$ (the size of the register may increase during the computation as we allow appending new qubits to it).
  \end{itemize}
We represent a basis state in the corresponding Hilbert space as $\ket{x,p,w}_{\mc{A}}$, where $\mc{A} = \mc{XPW}$ is a shorthand for the registers on which the algorithm operates. The initial state of the memory is the all-zero basis state $\ket{0, 0, 0}_{\mc{A}}$.

\paragraph{Quantum Phase Oracle.} We define the quantum oracle $\oq^D$ as the unitary operator acting on the memory of the algorithm as follows.
  \[\oq^D : \ket{x,p,w}_{\mc{A}} \mapsto \omega_N^{p D(x)} \ket{x,p,w}_{\mc{A}} \quad \mbox{where} \quad \omega_N = e^{\frac{2{\bf i}\pi}{N}}.\]
Note that this oracle returns the value $D(x)$ in the phase but it is equivalent to the standard oracle that maps $\ket{x,p,w}_{\CA}$ to $\ket{x,p \oplus D(x),w}_{\CA}$ up to a unitary transformation.

\paragraph{Classical Oracle.} A classical oracle query is defined as a query to the standard oracle that maps $\ket{x,p,w}_{\CA}$ to $\ket{x,p \oplus D(x),w}_{\CA}$ followed by a measurement on the index register $\mc{X}$ and phase register $\mc{P}$. Since we are working with phase oracles for convenience, we define them in the following way, equivalent to the above up to a unitary transformation.

We add a \emph{history} register $\mc{H} = \mc{H}_1 \cdots \mc{H}_t$ where the $c$-th subregister $\mc{H}_c$ is used to purify the $c$-th classical query (there are at most $t$ queries in total) and stores a value in $([M] \times [N]) \cup \{\star\}$. The initial state of that register is $\ket{\star,\dots,\star}_{\mc{H}}$. The classical oracle~$\oc^D$ is defined as the unitary operator acting as follows
  \[\begin{array}{ccl}
    \oc^D & : & \phantom{\omega_N^{p D(x)}} \ket{x,p,w}_{\mc{A}} \ket{(x_1,y_1),\dots,(x_c,y_c),\star,\dots,\star}_{\mc{H}} \\[3mm]
     & \mapsto & \omega_N^{p D(x)} \ket{x,p,w}_{\mc{A}} \ket{(x_1,y_1),\dots,(x_c,y_c),(x,D(x)),\star,\dots,\star}_{\mc{H}}.
  \end{array}\]

Since we only care about a bounded number of $t$ queries, the above oracle can easily be made a unitary.
For convenience, we denote the list $((x_1,y_1),\dots,(x_c,y_c),\star,\dots,\star)$ by $H$ and we say $x \in H$ if and only if there exists $1 \leq i \leq c$ such that $x_i = x$. We use the following shorthand for appending a new pair $(x,y)$ to $H$.

\begin{definition}[$H_{x \la y}$]
    Given a history $H = ((x_1,y_1),\ldots,(x_c,y_c), \star,\ldots, \star)$ with at least one star entry, we define
      \[H_{x \la y} = ((x_1,y_1),\ldots,(x_c,y_c), (x,y), \star, \ldots, \star)\]
    where the leftmost star has been replaced with $(x,y)$.
\end{definition}

Sometimes, we will identify the above list with a function $H: [M] \to [N] \cup \{\star\}$ if there are no ambiguous pairs, i.e. no pairs of the form $(x,y)$ and $(x,y')$ where $y \neq y'$. We also let~$\history$ denote the set of all possible histories $H$.

\paragraph{Hybrid Oracle.} We extend the above definitions by allowing for probabilistic choices between the two oracles.
This is represented by a channel that applies the quantum oracle $\oq^D$ with probability $1-b$, for some $b \in [0,1]$, and applies the classical oracle $\oc^D$ otherwise. Additionally, we assume that the algorithm is provided with a query type bit (or ``flag'') indicating which oracle has been applied. We represent this operation by an isometry $\oqc^D$ acting as
  \[
    \oqc^D : \ket{x,p,w}_{\mc{A}} \ket{H}_{\mc{H}}
      \mapsto \omega_N^{p D(x)} \ket{x,p}_{\mc{XP}} \pt*{\sqrt{1-b} \cdot \ket{w0}_{\mc{W}} \ket{H}_{\mc{H}}
      + \sqrt{b} \cdot \ket{w1}_{\mc{H}} \ket{H_{x \la D(x)}}_{\mc{H}}}
   \]
where the bit appended to the workspace $w$ indicates the nature of the oracle.
We recover the quantum and classical oracles when $b = 0$ and $b = 1$ respectively (ignoring the query type bit). We will use $b \notin \rn$ in the analysis of noisy and bounded-depth quantum algorithms.

\paragraph{Hybrid Algorithm.} An algorithm with $t$ queries is defined as a sequence $U_0, \dots, U_t$ of unitary transformations acting on the memory register $\mc{A}$ and a list of real numbers
$\lb(1),\dots,\lb(t) \in [0,1]$ that specifies which interpolation parameter is used at each query.
The state $\ket{\psi^D_t}$ of the algorithm after $t$ queries is
  \begin{equation}
  \label{Eq:standD}
  \ket{\psi^D_t} = U_t \,\mc{O}_{\lb(t)}^D \, U_{t-1} \, \cdots \, U_1 \, \mc{O}_{\lb(1)}^D \, U_0 \, \pt{\ket{0}_{\mc{A}} \ket{\star,\dots,\star}_{\mc{H}}}.
  \end{equation}
The function~$D$ is chosen uniformly at random from the set $\set{D : [M] \ra [N]}$. We model that by adding another purification register (the \emph{database}) $\mc{D} = \mc{D}_0\dots\mc{D}_{M-1}$ where each subregister~$\mc{D}_x$ for $x \in [M]$ holds a value $D(x) \in [N]$ and we define the following joint state,
  \begin{equation}
  \label{Eq:stand}
    \ket{\psi_t} = \frac{1}{N^{M/2}} \sum_{D \in [N]^M} \ket{\psi^D_t}_{\mc{AH}} \otimes \ket{D}_{\mc{D}}
    = U_t \,\mc{O}_{\lb(t)} \, U_{t-1} \, \cdots \, U_1 \, \mc{O}_{\lb(1)} \, U_0 \, \ket{\psi_0},
  \end{equation}
where $\oqc := \sum_{D} \oqc^D \otimes \proj{D}_{\mc{D}}$ and $\ket{\psi_0} :=  \ket{0}_{\mc{A}} \otimes \ket{\star, \cdots, \star}_{\mc{H}} \otimes \frac{1}{N^{M/2}} \sum_{D} \ket{D}_{\mc{D}}$.

\paragraph{Output.} The output of a hybrid algorithm is obtained by performing a computational basis measurement on the final state $\ket{\psi_{t}}$ where we measure a designated part of the workspace register~$\CW$. Since in this paper the output is always a tuple $(x_1, \ldots, x_k) \in [M]^k$ with $k \leq 2$, by making $k$ extra classical queries, we may assume that all the indices $x_1, \ldots, x_k$ are in the history register at the end.

  \subsection{Models for NISQ Algorithms}
  We describe the three models of NISQ quantum query complexity than can be analyzed in our framework of hybrid algorithms and state some of their properties.

\paragraph{Model 1. Bounded Quantum Queries and Adaptiveness.}
We first consider the case of algorithms that make only two types of queries: quantum queries and classical queries (i.e.~$b \in \rn$). Here, one can consider two types of algorithms: static or adaptive. A ``static'' algorithm fixes the order of which type of queries to make before it interacts with the oracle. An ``adaptive'' algorithm adaptively chooses the query type for each individual query, as long as the total number of quantum (and classical) queries is unchanged.

Below, we present a theorem, as a special case of \cite[Theorem 1]{DFH22c}, showing that any hybrid algorithm can be assumed to be static without loss of generality.

\begin{theorem}\label{thm:adaptive}
  In the hybrid random oracle model, for any adaptive hybrid quantum algorithm making at most $q$ quantum queries and $c$ classical, there exists a static hybrid algorithm making at most $2 q$ quantum queries and $2 c$ classical queries such that their outputs are always identical.
\end{theorem}

Given the above theorem, we will only consider lower bounds for static algorithms in the rest of the paper.

Before we move on, we give some intuition on why \Cref{thm:adaptive} holds. For fixed $c, q$, there exists a sequence $\lb^* = \lb^*_1 \lb^*_2 \cdots \lb^*_{2c+2q} \in \{0, 1\}^{2c+2q}$ with exactly $2 c$ elements being $1$, such that every $\lb = \lb_1, \cdots, \lb_{c+q} \in \{0, 1\}^{c+q}$  is a subsequence of $\lb^*$. This was proved in \cite[Lemma 1]{DFH22c}; we ignore the proof and refer interested readers to \cite{DFH22c} for full details. Assuming the statement about the existence of such a sequence is true, a static hybrid algorithm just picks the fixed sequence $\lb^*$ and every time it makes the next query, it checks if the current query type in $\lb^*$ is equal to the next query type in $\lb$. If yes, it makes the query; otherwise, it makes a junk query (for example, regardless of the query type, querying on input $0$ classically and discarding both the input and output). This strategy results in identical behavior of the static hybrid algorithm and any adaptive hybrid algorithm.

\paragraph{Model 2. Noisy Quantum Queries.}
We next consider the case of algorithms that have access to a noisy quantum oracle with noise level $b \in [0,1]$. We define this model using the mixed state representation~$\rho$ of the memory of an algorithm (over the registers $\mc{XPW}$) and the channel $\mathcal{N}_{\mc{XP}}$ that dephases the index and phase registers (i.e. $\mathcal{N}_{\mc{XP}}(\rho) = \sum_{x,p} \pt*{\proj{x,p}\otimes\id_{\mc{W}}} \rho \pt*{\proj{x,p}\otimes\id_{\mc{W}}}$). The noisy oracle is represented by the channel
  \begin{equation}
    \label{Eq:noise}
    \mathcal{N}_b^D : \rho \mapsto (1-b)\cdot \oq^D\rho\oq^D\otimes\proj{0} + b\cdot \mathcal{N}_{\mc{XP}}\pt*{\oq^D\rho\oq^D}\otimes\proj{1}.
  \end{equation}
This channel dephases the query registers after each quantum query with probability $b \in [0,1]$ and appends a ``noise flag'' qubit indicating whether the dephasing occurred. The state of the algorithm after~$t$ queries is defined recursively as $\rho_0^D = \proj{0,0,0}$ and $\rho_{t}^D = U_{t} \mathcal{N}_b^D(\rho_{t-1}^D) U_{t}^{\dagger}$ where $U_{t}$ is the unitary operator applied by the algorithm after the $t$-th query. One can observe that the hybrid oracle $\oqc^D$ is a purification of the noise channel $\mathcal{N}_b^D$, where the environment is enacted by the history register $\mc{H}$.

\begin{fact}
  Let $\ket{\psi_t^D}$ be the state defined in \Cref{Eq:standD} for a given sequence of unitaries $U_0,\dots,U_t$ and hybrid oracles $\mc{O}_{\lb(1)}^D,\dots,\mc{O}_{\lb(t)}^D$. Let $\rho_t^D$ be the state obtained by applying the same sequence of unitaries and replacing each oracle $\mc{O}_{\lb(i)}^D$ with $\mathcal{N}_{\lb(i)}^D$. Then, $\rho_t^D = \Tr_{\mc{H}}(\proj{\psi_t^D})$.
\end{fact}

This fact implies that the complexity of solving any problem using noisy quantum oracles is captured by the above model of hybrid algorithms. We will use this connection to derive the complexity of the preimage search and collision finding problems with noisy oracles.

Notice that our model is particularly versatile for proving hardness results (as is the goal in the present paper). Indeed, it can simulate algorithms that do not have access to the noise flag (just ignore the flag), algorithms that are subject to depolarizing noise (measure the flag qubit and depolarize the state on purpose when it is 1) and algorithms whose entire memory is subject to noise. Hence, our lower bounds apply to these models as well.

\paragraph{Model 3. Bounded Quantum Depth.}
Finally, we consider the model of bounded-depth quantum computation where the entire system decoheres periodically. Given a depth parameter~$d$, this amounts to applying the channel $\mathcal{N}_{\mc{XPW}}$ that dephases all the memory (i.e. $\mathcal{N}_{\mc{XPW}}(\rho) = \sum_{x,p,w} \proj{x,p,w} \rho \proj{x,p,w}$) every $d$ queries. The state of the algorithm can again be defined recursively as $\rho_0^D = \proj{0,0,0}$, $\rho_{t}^D = U_{t} \mathcal{N}_{\mc{XPW}}\pt*{\oq^D \rho_{t-1}^D \oq^D} U_{t}^{\dagger}$ if $t$ is a multiple of $d$, and $\rho_{t}^D = U_{t} \oq^D \rho_{t-1}^D \oq^D U_{t}^{\dagger}$ otherwise. This captures the scenario where a classical computer has access to a quantum computer of depth~$d$ and performs $t$ queries in total, which is also known as the $d$-CQ scheme~\cite{CCL23j,CM20c}.

We show that any $d$-depth algorithm can be simulated by an unbounded-depth algorithm that uses the hybrid oracle $\mc{O}_{1/d}$ without increasing the query complexity significantly. Intuitively, the interpolation parameter $1/d$ is sufficiently small so that~$d$ calls to the hybrid oracle will behave almost as $d$ calls to the quantum oracle.

\begin{proposition}
  \label{Prop:reduction}
  Fix any $d$-depth algorithm that makes $t$ quantum queries in total. Then, there exists an algorithm in the hybrid model that makes at most $2t$ queries \emph{in expectation} to the oracle~$\mc{O}_{1/d}$ and outputs the same outcome as the bounded-depth algorithm.
\end{proposition}

\begin{proof}
  It is sufficient to explain how to simulate a sequence of $d$ quantum queries using at most~$2d$ queries in expectation to the hybrid oracle~$\mc{O}_{1/d}$. The proposition follows by applying this simulation to the $\lceil t/d\rceil$ sequences of queries occurring between the applications of the channel~$\mathcal{N}_{\mc{XPW}}$ in the bounded-depth model.

  Consider an algorithm making $d$ queries to a quantum oracle $\mc{O}_{0}^D$. Suppose that we instead use the hybrid oracle $\mc{O}_{1/d}^D$ and measure after each query whether the query type bit is $0$ -- indicating that the query is quantum. If it is not $0$, we restart the simulation (the initial memory is classical, hence it can be cloned to restart as many times as needed). The algorithm stops once it obtains a sequence of $d$ consecutive $0$ (which will perfectly simulate the bounded-depth algorithm). Since each query is quantum with probability $1-1/d$, the expected number of calls to $\mc{O}_{1/d}^D$ corresponds to the number of coin flips needed to get $d$ consecutive heads when a coin has probability $1-1/d$ of coming up heads. This is equal to $((1-1/d)^{-d}-1)d \leq 2d$.
\end{proof}

We can easily modify the above algorithm to make exactly $4t$ queries to~$\mc{O}_{1/d}$ and succeeds in doing the simulation with probability at least $1/2$. This leads to the following corollary for deriving lower bounds in the depth-bounded model.

\begin{corollary}
  \label{Cor:reduction}
  Let $\sigma(t,b)$ denote the optimal success probability for solving a given problem using~$t$ queries to the oracle~$\mc{O}_{b}$ where $b \in (0,1]$. Then, the optimal success probability for solving the same problem using $t$ quantum queries in the bounded-depth model with depth $d = \lceil1/b\rceil$ is at most $2\sigma(4t,b)$.
\end{corollary}

While this reduction may not be tight in general, we show in this paper that it provides optimal bounds (up to constant factors) for the preimage search and collision finding problems.

\section{Hybrid Compressed Oracle}
\label{sec:compressed_oracle_for_hybrid}
In this section, we define the hybrid compressed oracle framework and prove some of its main properties. We also describe general results for constructing and analyzing progress measures in this framework.

\subsection{Construction}
\label{Sec:construction}

We start by defining the compressed encoding of the database that will be compatible with the history register. For this, we first augment the alphabet used for the database register such that~$\mc{D}_x$ can now hold $D(x) \in \set{\bot} \cup [N]$ and with the convention that $\omega_N^{p D(x)} = 1$ if $D(x) = \bot$. The initial state of the database is defined to be $\ket{\bot,\dots,\bot}_{\mc{D}}$. We also augment the alphabet of the history register so it can also store tuples of the form $(x, \bot)$ where $x \in [M]$. We say that $x \in H$ if there is a tuple of the form $(x,y) \in H$ where $y \in \{\bot\} \cup [N]$. Note that if there are no ambiguous pairs in the list, we can identify~$H$ as a function mapping $[M]$ to $\{\bot, \star\} \cup [N]$ with the extended alphabet (we will prove in Proposition~\ref{Prop:consistency} that such a property always holds in practice).

Next, we define the uncompression operator $S$. Let
  $\ket{\phat}_{\CD_x} = \frac{1}{\sqrt{N}} \sum_{y \in [N]} \omega_N^{p y} \ket{y}_{\mc{D}_x}$ for $p = 0,\dots,N-1$,
denote the Fourier basis states and let $S_x$ be the unitary operator acting on $\mc{D}_x$ such that
  \[S_x :
    \left\{
        \begin{array}{ll}
            \ket{\bot}_{\mc{D}_x} & \longmapsto~ \ket{\hat{0}}_{\CD_x}  \\[10pt]
            \ket{\hat{0}}_{\CD_x}  & \longmapsto~ \ket{\bot}_{\mc{D}_x} \\[10pt]
            \ket{\phat}_{\CD_x}  & \longmapsto~ \ket{\phat}_{\CD_x}  \quad \mbox{for $p = 1,\dots,N-1$.}
        \end{array}
    \right.
  \]
Note that $S_x$ is unitary and Hermitian. We now define a controlled unitary $S_{x, H}$ acting on $\mc{D}_x$:
  \begin{equation}
  \label{eqn:uncomp}
    S_{x, H} = \begin{cases}
         \mathbb{I} & \text{ if } x \in H\\
         S_x & \text{ otherwise}.
    \end{cases}
  \end{equation}
Define the Hermitian unitary operator $S$ acting on $\mc{AHD}$ such that:
  \[S = \sum_{x \in [M],  H \in \history} \proj{x}_{\mc{X}} \otimes \id_{\mc{PW}} \otimes \proj{H}_{\mc{H}} \otimes \pt{\id_{\mc{D}_0 \dots \mc{D}_{x-1}} \otimes S_{x, H} \otimes \id_{\mc{D}_{x+1} \dots \mc{D}_{M-1}}}.\]
The hybrid compressed oracle $\qc$ is defined as follows,
  \[\qc = S \oqc S  \quad \mbox{where} \quad \oqc = \sum_{D \in \pt{\set{\bot} \cup [N]}^M} \oqc^D \otimes \proj{D}_{\mc{D}},\]
for $b \in [0,1]$. The idea behind these definitions is that, for any basis state $\ket{x, p, w}_{\mc{A}} \ket {H,D}_{\mc{HD}}$:
  \begin{itemize}
      \item If the queried input satisfies $x \in H$, it means that $x$ has been queried classically before; then we stop (un)compressing $\mc{D}_x$, and it behaves like a regular phase oracle on input $x$.
      \item Otherwise $x \not\in H$, then $\mc{D}_x$ is simulated as a compressed oracle.
  \end{itemize}
In particular, note that the quantum compressed oracle $\qq$ only acts on the register $\CH$ as control. We provide an alternative definition to $\qq$ and $\cc$ in \Cref{Sec:sampl} that makes these observations more formal. Finally, the joint state $\ket{\phi_t}$ of the algorithm and the oracle after $t$ queries in the compressed oracle model is defined as
  \begin{equation}
    \label{Eq:comp}
    \ket{\phi_t} = U_t \, \mc{R}_{\lb(t)} \, U_{t-1} \, \cdots \, U_1 \, \mc{R}_{\lb(1)}  \, U_0 \, \pt{\ket{0}_{\mc{A}} \ket{\star,\dots,\star}_{\mc{H}} \ket{\bot,\dots,\bot}_{\mc{D}}}.
  \end{equation}
Following from \Cref{Eq:comp}, we define the initial state $\ket{\phi_0} =  \ket{0}_{\mc{A}} \otimes \ket{\star, \cdots, \star}_{\mc{H}} \otimes \ket{\bot,\dots,\bot}_{\mc{D}}$.

%%%%%%%%%%%%%%%%%%%%%%%%%%%%%%%%%%%%%%%%%%%%%%%%%%%%%%%%%%%%%%%%%%%%%%%%%%%%%%%%

\subsection{Structural Properties}

\paragraph{Indistinguishability.}

We show that the compression and uncompression operations behave as intended. For this, we will need some auxiliary definitions and lemmas. Let us define the unitary operator $\Sall$ that applies $S_{x, H}$ on every $\mc{D}_x$:
  \[\Sall = \sum_{H \in \history} \id_{\mc{XPW}} \otimes \proj{H}_{\mc{H}} \otimes \pt{S_{0, H} \otimes S_{2, H} \otimes \cdots \otimes S_{M-1, H}}.\]
In other words, we uncompress every entry of $\mc{D}$ (that is not in $H$) instead of only $\mc{D}_x$.
Observe that $\Sall \ket {\phi_0} = \ket{\psi_0}$.
We also have the following proposition:

\begin{proposition}
  \label{Prop:equivalent_uncompress}
  $\qc = S \oqc S = \Sall\oqc  \Sall$ for all $b \in [0,1]$.
\end{proposition}

\begin{proof}
  This is because for $\ket{x, p, w}_{\mc{A}}$, the oracle $\oqc$ acts as identity on the registers $\mc{D}_{<x}$ and~$\mc{D}_{>x}$. Therefore, for every $x' \ne x$, we have that $S_{x'}$ in the left multiplication with $\Sall$ cancels with $S_{x'}$ in the right multiplication with $\Sall$.
\end{proof}

The next proposition shows that $\ket{\phi_t}$ in the compressed oracle framework can be viewed as a compressed encoding of the state $\ket{\psi_t}$.

\begin{proposition}[Indistinguishability]
  \label{Prop:uncomp}
  The states $\ket{\psi_t}$ from \eqref{Eq:stand} and $\ket{\phi_t}$ from \eqref{Eq:comp} satisfy
    $\Sall \ket{\phi_t} = \ket{\psi_t}$.
  In particular, the two states are identical when we trace out the database register.
\end{proposition}

\begin{proof}
Using \eqref{Eq:comp}, the left-hand side is equal to
  \begin{align*}
    \Sall \, \ket{\phi_t} &= \Sall \, U_t \, \mc{R}_{\lb(t)} \, U_{t-1} \, \cdots \, U_1 \, \mc{R}_{\lb(1)}  \, U_0 \, \ket{\phi_0} \\
    &= \Sall \, U_t \,  (\Sall \mc{O}_{\lb(t)} \Sall)\, U_{t-1} \, \cdots \, U_1 \,  (\Sall \mc{O}_{\lb(1)} \Sall) \, U_0 \, \ket{\phi_0} \\
    &= (\Sall \, \Sall) \, U_t \,  \mc{O}_{\lb(t)} (\Sall \, \Sall) \, U_{t-1} \, \cdots \, U_1 \,  \mc{O}_{\lb(1)}  \, U_0 \, \Sall \, \ket{\phi_0}  \\
%    &= U_t \,  \mc{O}_{\lb(t)} \, U_{t-1} \, \cdots \, U_1 \,  \mc{O}_{\lb(1)}  \, U_0 \, \Sall \, \ket{\phi_0}  \\
    &= U_t \,  \mc{O}_{\lb(t)} \, U_{t-1} \, \cdots \, U_1 \,  \mc{O}_{\lb(1)}  \, U_0 \,   \ket{\psi_0} \\
    &= \ket{\psi_t}.
  \end{align*}
The second line follows from \Cref{Prop:equivalent_uncompress}. The third line is true because $U_i$ only operates on~$\mc{A}$ and commutes with $\Sall$ (which only operates on $\mc{HD}$). Finally, the last line uses that $\Sall$ is Hermitian, unitary and satisfies $\Sall \ket {\phi_0} = \ket{\psi_0}$.
\end{proof}

\paragraph{Consistency.}

We aim at characterizing what basis states can be in the support of~$\ket{\phi_t}$. We introduce the following vector space $\alg$ spanned by \emph{consistent states}.

\begin{definition}[History-Database Consistent State]
  \label{Def:consistent}
  Given an integer $t$, we say that $(H,D)$ is a \emph{history-database $t$-consistent pair} it it has the following properties:
      \begin{enumerate}
         \item (\textsc{Database size}) The database satisfies $D(x) \neq \bot$ for at most $t$ different values of $x$.
         \item (\textsc{History size}) The history is of the form $H = ((x_1,y_1),\dots,(x_c,y_c),\star,\dots,\star)$ where $x_1,\dots,x_c \in [M]$ and $y_1,\dots,y_c \in \set{\bot} \cup [N]$ for some $c \leq t$.
         \item (\textsc{Uniqueness}) We can identify the history with a function $H: [M] \ra \set{\star, \bot} \cup [N]$ where $H(x_j) = y_j$ for all $j \in \{1, 2, \cdots c\}$ (meaning no two pairs in the history can differ on the second coordinate only) and $H(x) = \star$ for $x \notin \set{x_1,\dots,x_c}$.
          \item (\textsc{Equality}) The database coincides with the history on non-$\star$ values, meaning that $H(x) \ne \star$ implies $D(x) = H(x)$.
      \end{enumerate}
  We let $\alg$ denote the vector space spanned by all basis state $\ket{x,p,w}_{\mc{A}}\ket{H,D}_{\mc{HD}}$ where~$(H,D)$ is history-database $t$-consistent. We say that a basis state is \emph{history-database consistent} if it is in~$\alg$ for some integer~$t$.
\end{definition}

The reader may wonder why we allow the history register to contain $(x,\bot)$ in the above definition since such a case shall not occur in $\ket{\psi_t}$ and $\ket{\phi_t}$ because of \Cref{Prop:uncomp}. This is only to provide more flexibility in further analysis. We now prove that $\kpt$ is supported over consistent basis states only.

\begin{proposition}[Consistency]
    \label{Prop:consistency}
    Any state $\ket{\phi_t}$ obtained after $t$ queries in the compressed oracle model satisfies $\ket{\phi_t} \in \alg$.
\end{proposition}

\begin{proof}
  We check the four properties stated in Definition~\ref{Def:consistent}. The first property follows from the fact that each query can increase the number of non-$\bot$ entries in $D$ by at most 1. For the second and third properties, we note that they hold for $\ket{\psi_t}$ and, by Proposition~\ref{Prop:uncomp}, the states~$\ket{\psi_t}$ and~$\ket{\phi_t}$ have the same reduced density matrix over $\mc{H}$. Finally, the fourth property holds for~$\ket{\psi_t}$ since $H(x) \ne \star$ implies that $D(x) = H(x)$. By Proposition~\ref{Prop:uncomp} and Equation~\eqref{eqn:uncomp}, for any $x$ such that $H(x) \ne \star$, the unitary $\Sall$ acts like an identity on $\mc{D}_x$. Therefore, the same holds for~$\ket{\phi_t}$ as well.
\end{proof}

Because of the above proposition, it suffices to only consider history-database consistent basis states while analyzing any algorithm and we shall tacitly assume that this is the case in any of the proofs that follow.

%%%%%%%%%%%%%%%%%%%%%%%%%%%%%%%%%%%%%%%%%%%%%%%%%%%%%%%%%%%%%%%%%%%%%%%%%%%%%%%%

\subsection{Sampling and Resampling}
\label{Sec:sampl}

In this section, we prove that the compressed oracle follows a similar behavior as the classical lazy-sampling strategy, namely the sampling of each input coordinate is delayed until it gets queried. There are some crucial differences yet, due to the reversibility of quantum computation. In particular, a coordinate can get ``resampled'' to a different value with a small probability.

In the rest of the paper, we abbreviate the root of unity $\omega_N = e^{\frac{2{\bf i}\pi}{N}}$ as~$\omega$. We also adopt the following notation to modify one entry of a database (we recall that for a history $H$ the notation~$H_{x \la y}$ is used for appending $(x,y)$ to the list).

\begin{definition}[$D_{x \la y}$]
  Let $(x,y) \in [M] \times (\set{\bot} \cup [N])$. Given $D : [M] \ra \set{\bot} \cup [N]$, we define the database $D_{x \la y}$ over the same domain as $D$ by
    \[
    D_{x \la y}(x') =
    \left\{
        \begin{array}{ll}
            y     & \text{if $x' = x$,}  \\[1pt]
            D(x') &\text{if $x' \neq x$.}
        \end{array}
    \right.
    \]
\end{definition}

The next lemmas describe what happens to the history and database when making a quantum or classical query. Among all the cases described below, the most interesting one is when the query is made at an index $x$ that is in the database but not in the history (i.e. $D(x) \neq \bot$ and $H(x) = \star$): up to a small resampling error, the database remains unchanged apart from an added phase.

\begin{lemma}[Quantum Query $\qq$]
  \label{Lem:recordQ}
  Let $\ket{x,p,w}\ket{H,D}$ be a history-database consistent basis state. Then, $\qq$ maps this state to $\ket{x,p,w0}\ket{H}\ket{\varphi}$ where the state $\ket{\varphi}$ of the database register is
  \begin{align*}
    \cdot\ & \omega^{p D(x)} \ket{D} \tag{\small if $H(x) \neq \star$ or $p = 0$} \\
    \cdot\ & \sum_{y \in [N]} \frac{\omega^{p y}}{\sqrt{N}} \ket{D_{x \la y}} \tag{\small if $H(x) = \star$, $D(x) = \bot$, $p \neq 0$} \\
    \cdot\ & \omega^{p D(x)} \ket{D} + \frac{\omega^{p D(x)}}{\sqrt{N}} \ket{D_{x \la \bot}} + \sum_{y \in [N]} \frac{1-\omega^{p D(x)} - \omega^{p y}}{N}  \ket{D_{x \la y}} \tag{\small if $H(x) = \star$, $D(x) \neq \bot$, $p \neq 0$}
  \end{align*}
\end{lemma}

\begin{lemma}[Classical Query $\cc$]
  \label{Lem:recordC}
  Let $\ket{x,p,w}\ket{H,D}$ be a history-database consistent basis state. Then, $\cc$ maps this state to $\ket{x,p,w1}\ket{\varphi}$ where the state $\ket{\varphi}$ of the history-database registers is
  \begin{align*}
    \cdot\ & \omega^{p D(x)} \ket{H_{x \la D(x)}, D} \tag{\small if $H(x) \neq \star$} \\
    \cdot\ & \sum_{y \in [N]} \frac{\omega^{p y}}{\sqrt{N}} \ket{H_{x \la y}, D_{x \la y}} \tag{\small if $H(x) = \star$, $D(x) = \bot$} \\
    \cdot\ & \omega^{p D(x)} \ket{H_{x \la D(x)}, D} + \frac{1}{\sqrt{N}} \ket{H_{x \la \bot}, D_{x \la \bot}} - \sum_{y \in [N]} \frac{\omega^{p y}}{N}  \ket{H_{x \la y}, D_{x \la y}} \tag{\small if $H(x) = \star$, $D(x) \neq \bot$}
  \end{align*}
\end{lemma}

In the above lemmas, when $x$ is not in the history but is in the database, after making a quantum or classical query, most likely $D(x)$ remains unchanged (corresponding to the $\ket{D}$ term), but there is a small probability that $D(x)$ gets removed (corresponding to the $\ket{D_{x \la \bot}}$ term) or resampled (corresponding to a superposition of $\ket{D_{x \la y}}$ over $y$).  We call the first term ``unchanged term'' (the database does not get updated), the second term ``removed term'' (the outcome on $x$ gets removed) and the last one ``resampled term'' in both items above. The proofs can be found in \Cref{sec:proof_resample}.

%%%%%%%%%%%%%%%%%%%%%%%%%%%%%%%%%%%%%%%%%%%%%%%%%%%%%%%%%%%%%%%%%%%%%%%%%%%%%%%%

\subsection{Progress Measures}

All progress measures studied in this paper will be expressed in terms of the norm of the projection onto basis states satisfying certain predicates.

\begin{definition}[Basis-State Predicate]
  Let $\predP : (x,p,w,H,D) \mapsto \set{\predF,\predT}$ be a predicate function over all basis states $\ket{x,p,w}_{\CA}\ket{H,D}_{\mc{HD}}$. We define the projection
    \[\pp = \sum_{(x,p,w,H,D) \in \predP^{-1}(\predT)} \proj{x,p,w,H,D}\]
  over all basis states satisfying \predP. We let~$\br{\predP}$ denote the \emph{negation} of~\predP\ and, given two predicates $\predP_1$ and $\predP_2$, we let $\predP_1 \cdot \predP_2$ denote their \emph{conjunction} and $\predP_1 + \predP_2$ denote their \emph{disjunction}.
\end{definition}

\begin{fact}
  \label{Fact:pred}
  Let $\predP_1$ and $\predP_2$ be two basis-state predicates. Then, the projections $\Pi_{\predP_1}$ and $\Pi_{\predP_2}$ are commuting operators. We have $\Pi_{\br{\predP_1}} = \id - \Pi_{\predP_1}$, $\Pi_{\predP_1 \cdot \predP_2} = \Pi_{\predP_1}\Pi_{\predP_2}$ and $\Pi_{\predP_1 + \predP_2} = \Pi_{\predP_1} + \Pi_{\predP_2} - \Pi_{\predP_1}\Pi_{\predP_2}$. Moreover, $P_1 \Rightarrow P_2$ if and only if $\Pi_{\predP_1} \preceq \Pi_{\predP_2}$, where $\preceq$ is the Loewner order.
\end{fact}

Most of the predicates considered in this paper will in fact depend only on the values of~$H$ and $D$ (a few predicates will also depend on the query index $x$).

We define the following general notions of progress measure and overlap.

\begin{definition}[Progress Measure and Progress Overlap]
  \label{Def:progIncr}
  Given a state $\kp$, a real $b \in [0,1]$ and a projector $\Pi$ over~$\mc{AHD}$, we define
    \[\quad \Delta_b(\Pi,\kp) = \norm{\Pi \qc \kp}^2 - \norm{\Pi \kp}^2 \quad \mathrm{and} \quad
    \Gamma_b(\Pi,\kp) = \frac{\norm{\Pi \qc (\id - \Pi) \kp}^2}{\norm{(\id - \Pi) \kp}^2}, \]
  with the convention that $\Gamma_b(\Pi,\kp) = 0$ if $\norm{(\id - \Pi) \kp} = 0$.
\end{definition}

The quantity $\Delta_b(\Pi,\kp) \in [-1,1]$ represents the increase in norm of the projection onto $\Pi$ after applying a hybrid query $\qc$. These will be used as a measure of progress later in the proofs.

The quantity $\Gamma_b(\Pi,\kp) \in [0,1]$ tracks the amplitude that moves after making a query from a subspace to its orthogonal complement. In particular, if $\Gamma_b(\Pi,\kp) \le \gamma$, then we have that $\norm{\Pi \qc (\id - \Pi) \kp}^2 \le \gamma \norm{(\id - \Pi) \kp}^2$. In this paper, we only consider projectors $\pp$ for some predicates $\predP$. In such cases, we can equivalently write
\[\Gamma_b(\pp,\kp) = \frac{\norm{\pp \qc \pnp \kp}^2}{\norm{\pnp \kp}^2}.\]

Next, we give two general lemmas that bound how much increase a single classical or quantum query can have towards a target history--database pair. These lemmas will apply when the predicate satisfies the following definition, which is similar to the notion of ``database property'' introduced in \cite{CMS19c,CFHL21c}. One difference in our definition is that we need to take the classical history into account.

\begin{definition}[History-Database Predicate]
  \label{Def:hdPred}
  Let $\predP : (H,D) \mapsto \set{\predF,\predT}$ be a predicate function over all history-database pairs. We say that it is a \emph{history-database predicate} if for every true-pair $(H,D) \in \predP^{-1}(\predT)$,
  \begin{itemize}
    \item (\textsc{Consistent}) The pair $(H,D)$ is history-database consistent (see Definition~\ref{Def:consistent}).
    \item (\textsc{History Invariant}) For every list $H'$ such that $(H',D)$ is history-database consistent and $H(x') = H'(x')$ for all $x' \in [M]$, we have $(H',D) \in \predP^{-1}(\predT)$.
    \item (\textsc{Database Monotone}) For every database $D'$ that is obtained by replacing a $\bot$ in $D$ with another value (i.e. $D = D'_{x' \la \bot}$ for some $x' \in [M]$), we have $(H,D') \in \predP^{-1}(\predT)$.
  \end{itemize}
  By extension, we say that $\predP : (x,p,w,H,D) \mapsto \set{\predF,\predT}$ is a history-database predicate if it does not depend on $(x,p,w)$ and its restriction to $(H,D)$ satisfies the above properties.
\end{definition}

The next lemmas bound the progress overlap $\Gamma_0$ (resp. $\Gamma_1$) in terms of the probability $\gamma$ that a history-database predicate becomes true when a new uniformly random value $y$ is added to the database (resp. database and history). We first provide the lemma for quantum queries, which follows the ideas used in previous work, starting from \cite{Zha19c}. Then we state the lemma for classical queries, which is new, but the core argument in the proof is similar. These results encompass most, although not all (see Lemma~\ref{Lem:samplH}), of the progress overlap bounds needed in subsequent applications. The proofs can be found in \Cref{sec:proof_progress_measure}.

\begin{lemma}[Progress Overlap, Quantum Query]
  \label{lem:quantum_query_progress}
  Let $\predP$ be a history-database predicate, $t$ be an integer and $\gamma \in [0, 1]$ be a real parameter. Suppose that, for every false-state $(H,D) \in \predP^{-1}(\predF) \cap \alg$ where $D(x) = \bot$, the probability to make the predicate true by replacing $D(x)$ with a random value $y$ is at most
  \begin{equation}
    \label{Eq:predQ}
    \Pr_{y \gets [N]}\bc*{(H,D_{x \la y}) \in \predP^{-1}(\predT)} \leq \gamma.
  \end{equation}
  Then, the quantum progress overlap is at most $\Gamma_0(\pp,\kp) \leq 10 \gamma$ for all~$\kp \in \alg$.
\end{lemma}

The adaptation of the above lemma to the classical query case requires making one extra assumption stated in Equation~(\ref{Eq:predC2}) below. This condition rules out predicates that can become true by simply copying a value from the database to the history.

\begin{lemma}[Progress Overlap, Classical Query]
  \label{lem:classical_query_progress}
  Let $\predP$ be a history-database predicate, $t$ be an integer and $\gamma \in [0, 1]$ be a real parameter. Suppose that, for every false-state $(H,D) \in \predP^{-1}(\predF) \cap \alg$ where $D(x) = \bot$, the probability to make the predicate true by replacing $H(x)$ and $D(x)$ with the same random value $y$ is at most
    \begin{equation}
      \label{Eq:predC1}
      \Pr_{y \gets [N]}\bc*{\pt*{H_{x \la y},D_{x \la y}} \in \predP^{-1}(\predT)} \leq \gamma.
    \end{equation}
  Assume further that, for every false-state $(H,D) \in \predP^{-1}(\predF)$, the predicate does not become true when $(x,D(x))$ is appended to the history, i.e.
    \begin{equation}
      \label{Eq:predC2}
      (H,D) \in \predP^{-1}(\predF) \quad \Rightarrow \quad (H_{x \la D(x)},D) \in \predP^{-1}(\predF).
    \end{equation}
  Then, the classical progress overlap is at most $\Gamma_1(\pp,\kp)  \leq 2 \gamma$ for all~$\kp \in \alg$.
\end{lemma}

Note that $\gamma$ will often depend on the maximum number $t$ of values contained in the database and in the history. Moreover, if \Cref{lem:quantum_query_progress,lem:classical_query_progress} hold with parameters~$\gamma_0$ and $\gamma_1$ respectively, then the progress interpolates as $\Gamma_b(\pp,\kp) \leq 10(1-b) \gamma_0 + 2b \gamma_1$.

Finally, we state some simple facts that will be used frequently throughout the paper.

\begin{fact}
  \label{Fact:inner_product_subsume}
    Let $\kp, \ket{\phi'}$ be two states defined over the registers $\mc{AHD}$. Let $U$ be a unitary operator over $\mc{A}$. Let $\Pi, \Pi'$ be two projectors over $\mc{AHD}$. Then,
    \begin{itemize}
        \item (Monotonicity) If $\Pi \preceq \Pi'$ then $\Pi \cdot \Pi' = \Pi' \cdot \Pi = \Pi$.
        \item (Commutativity) If $\Pi = \id_{\mc{A}} \otimes \Pi_{\mc{HD}}$ for some projector $\Pi_{\mc{HD}}$ then $\norm{\Pi U \kp} = \norm{\Pi \kp}$.
        \item (Sub-multiplicativity) $\norm{\Pi \kp} \leq \norm{\kp}$.
    \end{itemize}
\end{fact}

\section{Preimage Search}
\label{sec:hybrid_search}

In this section, we prove the lower bound for preimage search against hybrid algorithms.

\begin{theorem}
  \label{Thm:qc-search}
  The success probability of finding a zero preimage, in a uniformly random function $D:[M] \to [N]$, is at most
  \begin{itemize}
    \item \emph{(Model 1.)} $\bo[\big]{\frac{c + q^2}{N}}$ using $q$ quantum queries and $c$ classical queries,
    \item \emph{(Model 2.)} $\bo[\big]{\frac{t}{bN}}$ using $t$ queries to the hybrid oracle $\oqc$ where $1/t \leq b \leq 1$,
    \item \emph{(Model 3.)} $\bo[\big]{\frac{dt}{N}}$ using $t$ quantum queries with bounded-depth $1 \leq d \leq t$.
  \end{itemize}
\end{theorem}

The above inequalities are optimal. The proof proceeds as mentioned in the technical overview; we will define a notion of quantum and classical progress to keep track of the success probability of the algorithm after each query. To formally define these measures, we now give a series of predicates that characterize whether the history or the database contains a zero preimage:

\begin{definition}
  The following predicates evaluate a basis state $\ket{x,p,w,H,D}$ to $\predT$ if and only if it is history-database consistent (see \Cref{Def:consistent}) and satisfies the next conditions:
  \begin{itemize}
    \item \predQ: there exists a zero preimage in the quantum database $D$ that is not in the history $H$, i.e.~$x'$ such that $D(x') = 0$ and $H(x') = \star$.
    \item \predC: there exists a zero preimage in the classical history $H$, i.e.~$x'$ such that $H(x') = 0$. Note that for any history-database consistent basis state $(x',y) \in H$ implies $D(x')=y$, and thus if $\predC$ is true, then there exists~$x'$ such that $D(x') = H(x') = 0$.
    \item \predX\predQ: the predicate \predQ\ holds \emph{and} the query index $x$ is the only zero preimage in the quantum database $D$ that is not in the history $H$, i.e.~$D(x)=0$, $H(x) = \star$ and $H(x') = 0$ if~$D(x') = 0$ for all~$x' \neq x$.
    \item $\br{\predX}\predQ$: the predicate \predQ\ holds but not \predX\predQ\ (i.e., there exists~$x' \neq x$ such that $D(x') = 0$ and~$H(x') = \star$).
  \end{itemize}
\end{definition}

We shall also use negations, conjunctions and disjunctions of the above predicates.\\

To prove the lower bound, we first note that the squared norm $\norm{\pc \ket{\phi_t}}^2$ is an upper bound on the success probability of the algorithm after the last query since we can assume that the final output is always in the history register (by making one extra classical query at the end) and hence also in the database. We remark that because of the above, our hybrid compressed oracle framework avoids the need of using \cite[Lemma 5]{Zha19c} that is typically needed for proofs in the usual compressed oracle framework.

To keep track of the progress of the algorithm, we will need more fine-grained control and for this we keep track of the change in the quantities $\norm{\pc \ket{\phi_t}}$ and $\norm{\pqnc \ket{\phi_t}}$, which can be thought of as classical and (purely) quantum progress respectively. Initially, both quantities are equal to zero. Each time the algorithm makes a quantum ($b = 0$) or classical ($b = 1$) query, we show that the progress evolves as follows in terms of the quantity defined in \Cref{Def:progIncr}:
\begin{align*}
    \ \Delta_b(\Pi,\kp) &= \norm{\Pi \mc{R}_b \kp}^2 - \norm{\Pi \kp}^2.
\end{align*}

\begin{proposition}[Progress after a quantum query]\label{prop:quantum}
Given an integer $t$ and a state $\kp \in \alg$ with norm at most 1, the progress caused by one quantum query on $\kp$ are at most,
  \[\Delta_0(\pc,\kp)  = 0 \quad \text{ and }\quad
    \Delta_0(\pqnc,\kp) \le 2\sqrt{\frac{10}{N}}\norm{\pxqnc \kp} + \frac{10}{N}.\]
\end{proposition}

\begin{proposition}[Progress after a classical query]\label{prop:classical}
  Given an integer $t$ and a state $\kp \in \alg$ with norm at most 1, the progress caused by one quantum query on $\kp$ are at most,
  \begin{align*}
    \Delta_1(\pc,\kp) \le  2\norm{\pxqnc \kp}^2 + \delta + \frac{4}{N} \quad \text{ and } \quad    \Delta_1(\pqnc,\kp) = - \norm{\pxqnc \kp}^2 - \delta
  \end{align*}
where $\delta = \norm{\pc \cc \pnxqnc \kp}^2$.
\end{proposition}

The first proposition follows from a similar (but refined) analysis of the preimage search in the compressed oracle framework of \cite{Zha19c}. The second proposition is different from the usual analysis in this framework as it shows that a classical query can also decrease the progress the algorithm has made. We shall give their proofs later. First we show how the above imply optimal lower bounds for preimage search. It suffices to prove the results in models 1 and 2 since the result in model 3 follows by that in model 2 and \Cref{Cor:reduction}.

% \begin{theorem}
%   The progress made by any algorithm after $t = q+c$ queries, of which $q$ are quantum and $c$ are classical, satisfies
%     \[\norm{\pc \ket{\phi_t}}^2 = \bo*{\frac{c+q^2}{N}}.\]
% \end{theorem}

\begin{theorem}
  The progress made by any algorithm after $t$ queries satisfies
  \begin{itemize}
    \item \emph{(Model 1.)} $\norm{\pc \ket{\phi_t}}^2 = \bo[\big]{\frac{c+q^2}{N}}$ if $q$ queries use the quantum oracle and $c$ queries use the classical oracle with $c+q = t$,
    \item \emph{(Model 2.)} $\norm{\pc \ket{\phi_t}}^2 = \bo[\big]{\frac{t}{bN}}$ if all queries use the hybrid oracle $\oqc$ for some $1/t \leq b \leq 1$.
%    \item $O(dt^2/N)$ using $t$ queries to quantum oracle with depth $1 \leq d \leq t$.
  \end{itemize}
\end{theorem}

\begin{proof}
  It will be more convenient to keep track of the potential
  \[ \Psi_t := \norm{\pc \ket{\phi_t}}^2 + 3\norm{\pqnc \ket{\phi_t}}^2.\]
  Observe that $\norm{\pc \ket{\phi_t}}^2 \le \Psi_t$\footnote{In fact, since $\pc + \pqnc = \pqorc$ where the projectors in the sum are orthogonal, we also have that $\frac13 {\Psi_t} \le \norm{\pqorc \ket{\phi_t}}^2 \le \Psi_t$ but we do not use this fact.}.
  We claim that the following recurrence holds for the potential $\Psi_t$ if the $t$-th query is made to the oracle $\oqc$ with $b \in [0,1]$:
  \begin{equation}
    \label{eqn:recurrence}
    \Psi_t \le \Psi_{t-1} + \min\pt*{11\sqrt{\frac{\Psi_{t-1}}{N}},\ \frac{90}{bN}}+ \frac{30}{N}
   \end{equation}
  with the initial condition that $\Psi_0 = 0$. Recalling the definition of $\Delta_b(\Pi,\kp)$, it follows from the fact $\Psi_{t} = \norm{\pc\qc \kptt}^2 + 3\norm{\pqnc\qc \kptt}^2 = \Psi_{t-1} + \Delta_b(\pc,\kptt) + 3  \Delta_b(\phnc,\kptt)$ and \Cref{prop:quantum,prop:classical} that
  \begin{align*}
      \Psi_t
        & \le \Psi_{t-1} -b \norm{\pxqnc \kptt}^2 + 6(1-b)\sqrt{\frac{10}{N}}\norm{\pxqnc \kptt} - 2\delta b + \frac{30-26b}{N} \\
        & \leq \Psi_{t-1} - b \norm{\pxqnc \kptt}^2 + 6\sqrt{\frac{10}{N}}\norm{\pxqnc \kptt}  + \frac{30}{N}.
  \end{align*}
  We obtain \Cref{eqn:recurrence} by bounding the term $- b \norm{\pxqnc \kptt}^2 + 6\sqrt{\frac{10}{N}}\norm{\pxqnc \kptt}$ in two different ways: (1) using that $\norm{\pxqnc \kptt} \leq \norm{\pqnc \kptt} \leq \sqrt{\Psi_{t-1}/3}$, it is at most $11\sqrt{\Psi_{t-1}/N}$, (2) using that the polynomial $-b Z^2 + 6\sqrt{\frac{10}{N}} Z$ is maximized at $Z = \sqrt{\frac{90}{b^2N}}$ when $b > 0$, it is at most $\frac{90}{bN}$.

  If all queries use the same oracle $\oqc$, for some $b \neq 0$, then by \Cref{eqn:recurrence} we get
    \[\norm{\pc \ket{\phi_t}}^2 \le \Psi_t = \bo[\Big]{\dfrac{t}{bN}},\]
  which proves the second statement in the theorem.
  If instead each query is either to the quantum ($b = 0$) or classical ($b = 1$) oracle then the potential increases by at most $\Psi_t - \Psi_{t-1} \leq 11\sqrt{\Psi_{t-1}/N} + 30/N$ when making a quantum query and $\Psi_t - \Psi_{t-1} \leq 120/N$ when making a classical query by \Cref{eqn:recurrence}. We can assume that the last two inequalities are replaced with equalities as it can only increase the maximum possible value for $\Psi_t$. Observe in this case that the potential always increases by the same amount $120/N$ when making a classical query, whereas for quantum queries it is advantageous to first maximize the value of $\Psi_{t-1}$. Hence, for an algorithm making $c$ classical and $q$ quantum queries, the optimal strategy is to use the classical recurrence for the first $c$ steps and the quantum recurrence afterward. In this case, it follows that for $t=q+c$ queries, we have that
  \[\norm{\pc \ket{\phi_t}}^2 \le \Psi_t = \bo*{\dfrac{c+q^2}{N}}.\qedhere\]
\end{proof}

To complete the proof, we now prove \Cref{prop:quantum,prop:classical}.

\begin{proof}[Proof of \Cref{prop:quantum}]
  The first equality is due to the fact that a quantum query $\qq$ only uses the register $\CH$ as a control. Thus, for any basis state in the support of the projector $\pnc$, which does not contain a zero preimage in $H$ by definition, the state after applying~$\qq$ will still not contain a zero preimage in $H$ and thus be orthogonal to the support of $\pc$. On the other hand, a basis state in the support of $\pc$ contains a zero preimage in $H$ and remains in the support even after applying $\qq$. Since $\mathbb{I} = \pnc + \pc$ and the projectors in the summation are orthogonal, the statement $\norm{\pc \qq \ket{\phi}} = \norm{\pc \ket{\phi}}$ follows and hence $\Delta_0(\pc,\kp)=0$.

  To see the second inequality, we have that
    \begin{align}\label{eqn:preimage}
           \ \norm{\pqnc \qq \ket{\phi}}^2
             & = \norm{\pqnc \qq (\pc + \pqnc + \pnqnc)\ket{\phi}}^2 \notag \\
             & = \norm{\pqnc \qq (\pqnc + \pnqnc)\ket{\phi}}^2 \notag \\
             & = \norm{\pqnc \qq \pnxqnc\ket{\phi}}^2 + \norm{\pqnc \qq (\pxqnc + \pnqnc)\ket{\phi}}^2 \notag \\
             & \le \norm{\pnxqnc\ket{\phi}}^2 + \pt*{\norm{\pxqnc\ket{\phi}} + \norm{\pqnc \qq \pnqnc\ket{\phi}}}^2 \notag \\
             & = \norm{\pqnc\ket{\phi}}^2 + 2\norm{\pxqnc\ket{\phi}}\cdot \norm{\pqnc \qq \pnqnc\ket{\phi}} + \norm{\pqnc \qq \pnqnc\ket{\phi}}^2.
    \end{align}
  The second equality uses that $\pqnc \qq \pc = 0$ since any basis state in the support of $\pc$ will remain in the support of the same projector. This is because $\qq$ acts as a control on $\CH$ and there is already a zero preimage $x \in H$ before applying $\qq$. The third equality uses that $\pqnc \qq \pnxqnc\ket{\phi}$ is orthogonal to $\pqnc \qq (\pxqnc + \pnqnc)\ket{\phi}$ since the former is supported over basis states~$\ket{x}_{\mc{X}}\ket{D}_{\mc{D}}$ containing a zero preimage in $D$ not equal to $x$, whereas the latter can only contain the preimage $x$. The last two lines uses the triangle inequality and the fact that $\norm{\pnxqnc\ket{\phi}}^2 + \norm{\pxqnc\ket{\phi}}^2 = \norm{\pqnc\ket{\phi}}^2$.

  To bound the term $\norm{\pqnc \qq \pnqnc \ket{\phi}}$ in \eqref{eqn:preimage}, we use that $\pqnc = \pnc \cdot \pq$ and $\pnqnc = \pnq \cdot \pnc$ and thus, $\norm{\pqnc \qq \pnqnc \ket{\phi}} \le \norm{\pq \qq \pnq \pnc \ket{\phi}}$. Since $\predQ$ is a history-database predicate, we can apply \Cref{lem:quantum_query_progress} to bound the above by $\sqrt{\frac{10}{N}}\norm{ \pnc \ket{\phi}}$. Plugging this into \eqref{eqn:preimage} and rearranging, we get the desired inequality about $\Delta_0(\pqnc,\kp)$. \qedhere
\end{proof}

\begin{proof}[Proof of \Cref{prop:classical}]
  Towards proving the first inequality in the statement of the proposition, we have that
  \begin{align}\label{eqn:preimage2}
    \norm{\pc \cc \ket{\phi}}^2
      & = \norm{\pc \cc \pt*{\pc + \pqnc + \pnqnc} \ket{\phi}}^2 \notag \\
    \ & = \norm{\pc \cc \pc \kp}^2 + \norm{\pc \cc \pt*{\pqnc + \pnqnc} \kp}^2  \notag \\
    \ & = \norm{\pc \cc \pc \kp}^2 + \norm{\pc \cc \pnxqnc \ket{\phi}}^2  + \norm{\pc \cc \pt*{\pxqnc + \pnqnc} \kp}^2   \notag \\
    \ & \le \norm{\pc \ket{\phi}}^2 + \norm{\pc \cc \pnxqnc \kp}^2 + 2\norm{\pxqnc \kp}^2 + 2\norm{\pc \cc \pnqnc \kp}^2.
  \end{align}
  \begin{sloppypar}
  The second equality in the above sequence follows since $\pc \cc\pc \kp$ is orthogonal to $\pc \cc \pt*{\pqnc+\pnqnc} \kp$. This can be seen from the fact that the history register~$\mc{H}$ in $\pc \cc \pc \ket{\phi}$ is supported over basis states~$\ket{H}_{\mc{H}}$ where the first $c$ entries of $H$ contains a zero preimage. Therefore, it is orthogonal to $\pc \cc \pt*{\pqnc+\pnqnc} \kp$. Similarly, the third equality uses that $\pc \cc\pnxqnc \kp$ is orthogonal to $\pc \cc \pt*{\pxqnc+\pnqnc} \kp$ since the latter is supported over basis states~$\ket{x}_{\mc{X}}\ket{D}_{\mc{D}}$ where the only possible zero preimage in $D$ is $x$. The last inequality follows from \Cref{Fact:inner_product_subsume} and the fact that $\norm{\ket{a}+\ket{b}}^2 \le 2\norm{\ket{a}}^2+2\norm{\ket{b}}^2$ for any states $\ket{a}$ and~$\ket{b}$.
  \end{sloppypar}

  Since $\predQ+\predC$ is a history-database predicate satisfying Equation~(\ref{Eq:predC2}), we can use \Cref{lem:classical_query_progress} to bound the last term in \Cref{eqn:preimage2}. It gives us that for $\gamma = \frac1N$,
    \[\Gamma_1(\pqorc, \ket{\phi}) \le 2\gamma \implies \norm{\pqorc \cc \pnqnc \ket{\phi}}^2 \le 2\gamma \norm{\pnqnc \ket{\phi}}^2 \le \frac{2}{N}\]
  and $\norm{\pc \cc \pnqnc \kp}^2 \leq \norm{\pqorc \cc \pnqnc \ket{\phi}}^2$.
  Recalling \Cref{Def:progIncr}, we thus have shown that
    \[\Delta_1(\pc,\kp) \leq 2\norm{\pxqnc \ket{\phi_t}}^2 + \norm{\pc \cc \pnxqnc \kp}^2 + \frac{4}{N},\]
  proving the first inequality in the statement of the proposition.\\

  The second equality is relatively straightforward:
    \begin{align*}
     \norm{\pqnc \cc \ket{\phi}}^2
      & = \norm{\pqnc \cc \pt*{\pc + \pxqnc + \pnxqnc + \pnqnc} \ket{\phi}}^2  \\
      & = \norm{\pqnc \cc \pnxqnc \ket{\phi}}^2  \\
      & = \norm{\pnc \cc \pnxqnc \ket{\phi}}^2  \\
      & = \norm{\pnxqnc \ket{\phi}}^2 - \norm{\pc \cc \pnxqnc \ket{\phi}}^2  \\
      & = \norm{\pqnc \ket{\phi}}^2 - \norm{\pxqnc \ket{\phi}}^2 - \norm{\pc \cc \pnxqnc \ket{\phi}}^2,
    \end{align*}
  where the second line is true as $\pqnc \cc \pt*{\pc + \pxqnc + \pnqnc} = 0$. This holds since a classical query~$\cc$ can not remove a zero preimage from $H$ or lead to a zero preimage in the database that is not in the classical history as well. The third line follows since the state $\cc \pnxqnc \ket{\phi}$ must necessarily contain a zero preimage in the database (by the predicate $\br{\predX}\predQ$).
  The last two lines use that $\pc + \pnc = \id$ and $\pxqnc + \pnxqnc = \pqnc$ where the projectors in each sum are orthogonal.
\end{proof}

\section{Collision Finding}
\label{sec:hybrid_cols}

In this section, we prove our main theorem on hybrid collision-finding algorithms:

\begin{theorem}
  \label{Thm:qc-coll}
  The success probability of finding a colliding pair, in a uniformly random function $D:[M] \to [N]$, is at most
  \begin{itemize}
    \item \emph{(Model 1.)} $\bo[\big]{\frac{c^2 + c q^2 + q^3}{N}}$ using $q$ quantum queries and $c$ classical queries,
    \item \emph{(Model 2.)} $\bo[\big]{\frac{t^2}{bN}}$ using $t$ queries to the hybrid oracle $\oqc$ where $1/t \leq b \leq 1$,
    \item \emph{(Model 3.)} $\bo[\big]{\frac{dt^2}{N}}$ using $t$ quantum queries with bounded-depth $1 \leq d \leq t$.
  \end{itemize}
\end{theorem}

The section is organized as follows. The progress measures needed for the proof of the above theorem are introduced in \Cref{Sec:prog-coll}. The main part of the proof is contained in \Cref{Sec:main-coll}. It uses some auxiliary lemmas whose demonstrations are deferred to \Cref{Sec:samplLemCol,Sec:Ccoll}.

%%%%%%%%%%%%%%%%%%%%%%%%%%%%%%%%%%%%%%%%%%%%%%%%%%%%%%%%%%%%%%%%%%%%%%%%%%%%%%%%

\subsection{Progress Measure}
\label{Sec:prog-coll}

We define three types of collision pairs that can be recorded by a hybrid compressed oracle.

\begin{definition}[Collision Type]
  Given a history-database consistent pair $(H,D)$, we say that it contains a collision if there exist two values $x_1 \neq x_2$ such that $D(x_1) = D(x_2) \neq \bot$. Additionally, if $x_1,x_2 \notin H$ the collision is said to be \emph{quantum}, if $x_1,x_2 \in H$ it is said to be \emph{classical} and if $x_1 \notin H$, $x_2 \in H$ it is said to be \emph{hybrid}.
\end{definition}

We now give a series of predicates that characterize what types of collisions have been recorded in a basis state. Later on, we will combine these predicates together to define the measures of progress needed in our proofs.

\begin{definition}
  The following predicates evaluate a basis state $\ket{x,p,w,H,D}$ to $\predT$ if and only if it is history-database consistent (see \Cref{Def:consistent}) and satisfies the next conditions:
    \begin{itemize}
      \item $\predQ$, $\textsc{h}$, $\textsc{c}$: there is respectively at least one quantum, one hybrid or one classical collision contained in $(H,D)$.
      \item $\predX\predQ$: the predicate $\predQ$ holds \emph{and} the query index $x$ is contained in every quantum collision.
      \item $\predX\predH$: the predicate $\predH$ holds \emph{and} the query index is contained in every hybrid collision \emph{and} the query index is not in the history.
      \item $\br{\predX}\predQ$ (resp. $\br{\predX}\predH$): the predicate $\predQ$ (resp. $\predH$) holds, but not $\predX\predQ$ (resp. $\predX\predH$).
    \end{itemize}
\end{definition}

Note that $\predX\predQ + \br{\predX}\predQ = \predQ$ and $\predX\predH + \br{\predX}\predH = \predH$. Furthermore, the predicate $\br{\predX}\predQ$ is equivalent to the existence of a quantum collision not containing the query index. The last four predicates are the only ones that depend on the value~$x$ contained in the index register. The other predicates depend only on the history-database~$(H,D)$.

We will combine the above predicates into the potential
  \[\Psi(\kp) = \norm{\pc \kp}^2 + 3 \norm{\phnc \kp}^2 + 7 \norm{\pqnhnc \kp}^2\]
that allows for bounding the probability $\norm{\pqorhorc \kp}^2 = \norm{\pc \kp}^2 + \norm{\phnc \kp}^2 + \norm{\pqnhnc \kp}^2$ of recording any type of collision.

%%%%%%%%%%%%%%%%%%%%%%%%%%%%%%%%%%%%%%%%%%%%%%%%%%%%%%%%%%%%%%%%%%%%%%%%%%%%%%%%

\subsection{Main Result}
\label{Sec:main-coll}

We now turn to the proof of \Cref{Thm:qc-coll}, delaying auxiliary lemmas to later sections. First, it is simple to argue that, for a $t$-query algorithm computing a state $\kpt$ in the hybrid compressed oracle model, the probability $\norm{\pqorhorc \kpt}^2$ of recording any type of collision is an upper bound on the success probability.
Since a direct bound on this quantity is difficult to obtain, we instead analyze the three predicates $\predC$, $\predH \cdot \br{\predC}$, $\predQ \cdot \br{\predH} \cdot \br{\predC}$ separately, and later combine them into a bound on the potential~$\Psi(\kpt)$.

We first show that performing a quantum query incurs the following progress increases.

\begin{lemma}[Progress Measure, Quantum Query]\label{Lem:quProg}
  Given an integer $t$ and a state $\kp \in \alg$ with norm at most 1, the progress caused by one quantum query on $\kp$ are at most,
    \begin{align*}
        \Delta_0(\pc,\kp) & = 0 \, , \\
        \Delta_0(\phnc,\kp) & \leq 2\sqrt{\frac{10t}{N}}\norm{\pxhnc \kp} + \frac{10t}{N} \, , \\
        \Delta_0(\pqnhnc,\kp) & \leq \sqrt{\frac{8t}{N}} \norm{\pxhnc\kp} + 2\sqrt{\frac{20t}{N}} \norm{\pxqnhnc \kp} + \frac{20t}{N} \, .
    \end{align*}
\end{lemma}

Recall that, by \Cref{Def:progIncr}, the quantity $\Delta_0(\pp, \kp) \in [-1,1]$ for a predicate $\predP$ represents the progress increase $\Delta_0(\pp, \kp) = \norm{\pp \qq \kp}^2 - \norm{\pp \kp}^2$ when doing a quantum query. Hence, the first equality reflects the fact that a quantum query cannot create or destroy a classical collision. The second inequality is based on the observation that, when adding a random value to the database, the probability that it creates a hybrid collision is at most $t/N$ since it must collide with one of the at most $t$ values contained in the history. The third inequality is slightly more involved since it must also take into account the case of \emph{removing} a hybrid collision from the history-database.

We next look at the progress increase when the query is classical.

\begin{lemma}[Progress Measure, Classical Query]\label{Lem:clProg}
  Given an integer $t$ and a state $\kp \in \alg$ with norm at most 1, the progress caused by one classical query on $\kp$ are at most,
    \begin{align*}
      \Delta_1(\pc,\kp) & \leq 2 \norm{\pxhnc \kp}^2 + \delta_1 + \frac{4 t}{N} \, , \\
      \Delta_1(\phnc,\kp) & \leq - \norm{\pxhnc \kp}^2 + 2\norm{\pxqnhnc \kp}^2 - \delta_1 + 2\delta_2 + \frac{12t}{N} \, , \\
      \Delta_1(\pqnhnc,\kp) & \leq \sqrt{\frac{2t}{N}} \norm{\pxhnc\kp}  - \norm{\pxqnhnc \kp}^2 - \delta_2 \,
    \end{align*}
  where $\delta_1 = \norm{\pc \cc \pnxhnc \kp}^2$ and $\delta_2 = \norm{\phnc \cc \pnxqnhnc \kp}^2$.
\end{lemma}

The negative terms on the right-hand side represent the amount of progress transferred by one classical query between different progress measures. Note that, as a simple case, if an algorithm makes only classical queries then there can be no hybrid or quantum collision, hence $\norm{\pxhnc \kpt} = \norm{\pxqnhnc \kpt} = 0$ and the above inequalities simplify to $\Delta_1(\pc,\kpt) = \norm{\pc \cc \kpt}^2 - \norm{\pc \kpt}^2 \leq 4t/N$. Thus, we recover the birthday bound $\norm{\pc \kpt}^2 = \bo{t^2/N}$ after $t$ classical queries.

We now combine the two lemmas to bound the potential increase under applying the hybrid compressed oracle $\qc$.

\begin{proposition}
  \label{Prop:interpProg}
  Given an integer $t$ and a state $\kp \in \alg$ with norm at most 1, upon applying the hybrid compressed oracle $\qc$, the potential increases by at most
    \begin{equation}
      \Psi(\qc\kp) \leq \Psi\pt*{\kp} + \min\pt*{81\sqrt{\frac{t\cdot \Psi\pt*{\kp}}{N}},\ \frac{1641t}{bN}} + \frac{170t}{N}
    \end{equation}
  for all $b \in [0,1]$.
\end{proposition}

\begin{proof}
  We start by proving
  \begin{equation}
    \label{Eq:RecG}
    \Psi(\qc\kp) \leq \Psi\pt*{\kp} - b \norm{\pspec \kp}^2 + 81\sqrt{\frac{t}{N}} \norm{\pspec \kp} + \frac{170t}{N}
  \end{equation}
  for all $b \in [0,1]$. By combining \Cref{Lem:quProg,Lem:clProg} with the fact that
    \begin{align*}
      \Psi(\qc\kp) & = (1-b)\Psi(\qq\kp) + b\Psi(\cc\kp) \\
                   & = \Psi(\kp) + \Delta_b(\pc,\kp) + 3  \Delta_b(\phnc,\kp) + 7 \Delta_b(\pqnhnc,\kp) ,
    \end{align*}
    we have that
    \begin{align*}
    \Psi\pt*{\qc\kp} &\leq \Psi(\kp) - b (\norm{\pxhnc \kp}^2 + \norm{\pxqnhnc \kp}^2)\\
    \ &\hphantom{textttttttt}+ 40\sqrt{\frac{t}{N}} \norm{\pxhnc \kp} + 70\sqrt{\frac{t}{N}} \norm{\pxqnhnc \kp} + \frac{170t}{N}.
    \end{align*}
    \Cref{Eq:RecG} follows by observing that
    \begin{align*}
    \ \norm{\pxhnc \kp}^2 + \norm{\pxqnhnc \kp}^2 &= \norm{\pspec \kp}^2, \text{ and }\\
    \ 40\norm{\pxhnc \kp} + 70\norm{\pxqnhnc \kp} &\leq \sqrt{40^2+70^2} \norm{\pspec \kp},
    \end{align*}
    where the last inequality follows from Cauchy--Schwarz.

  Finally, the proposition is derived from \Cref{Eq:RecG} and the fact that
  \begin{align*}
   - b \norm{\pspec \kp}^2 &+ 81\sqrt{\frac{t}{N}} \norm{\pspec \kp}  \\
  \ &\hphantom{texttexttttt}\leq \min\set{81\sqrt{\frac{t}{N}} \norm{\pspec \kp}, \frac{1641t}{bN}} \\
  \ &\hphantom{texttexttttt} \leq \min\set{81\sqrt{\frac{t}{N} \cdot \Psi\pt*{\kp}}, \frac{1641 t}{bN}}
  \end{align*}
  since the polynomial $-b Z^2 + 81\sqrt{\frac{t}{N}} Z$ is maximized at $Z = 81\sqrt{\frac{t}{4b^2N}}$. \qedhere
\end{proof}

Finally, we can prove our main theorem by tuning the interpolation coefficient $b$.

\begin{proof}[Proof of \Cref{Thm:qc-coll}]
  We first consider the case of hybrid algorithms that only make classical or quantum queries (model 1). We want to upper bound the probability that an algorithm outputs a collision pair after $t = c + q$ queries, of which~$c$ are classical and~$q$ are quantum. Fix any such algorithm and let $\kpt$ denote its state as defined in \Cref{Eq:comp}.
  We can always assume, at the cost of doing two extra classical queries, that the output is contained in the history register. Hence, the success probability of the algorithm is upper bounded by the probability $\norm{\pc \ket{\phi_t}}^2$ of having recorded a classical collision. We now prove the upper bound $\norm{\pc \ket{\phi_t}}^2 = O((c^2 + c q^2 + q^3)/N)$ that matches our theorem. For that, we consider the potential after $t$ queries defined as
    \[\Psi_t := \norm{\pc \kpt}^2 + 3 \norm{\phnc \kpt}^2 + 7 \norm{\pqnhnc \kpt}^2.\]
  Our proof is by induction on $t$. Initially, $\Psi_0 = 0$ since the history and database registers of $\ket{\phi_0}$ are empty by definition. By \Cref{Prop:interpProg}, at each query, the potential increases by at most,
  \begin{equation*}
    \Psi_{t} \leq
      \begin{cases}
         \pt[\Big]{\sqrt{\Psi_{t-1}} + 41\sqrt{\frac{t-1}{N}}}^2 & \text{ if the $t$-th query is quantum $(b = 0),$}\\[2mm]
         \Psi_{t-1} + \frac{1811(t-1)}{N} & \text{ if the $t$-th query is classical $(b = 1)$}.
       \end{cases}
  \end{equation*}
  The maximum increase permitted by the above two inequalities is achieved when all the classical queries are performed first. Thus, we conclude that
    \[\Psi_{c+q} = \bo[\bigg]{c \cdot \frac{c+q}{N} + q^2 \cdot \frac{c+q}{N}} = \bo*{\frac{c^2+cq^2 + q^3}{N}}.\]

  We now study the case of algorithms that make $t$ queries to the same hybrid oracle $\oqc$ where~$b > 0$ (model 2). By using the same definition of $\Psi_t$ as above, together with \Cref{Prop:interpProg}, we obtain that
    \[\Psi_t = \bo*{\frac{t^2}{bN}}\]
  since each query increases the potential by at most $\bo{t/(bN)}$.

  Finally, the case of bounded-depth algorithms (model 3) follows by the result in model 2 and \Cref{Cor:reduction}.
\end{proof}

%%%%%%%%%%%%%%%%%%%%%%%%%%%%%%%%%%%%%%%%%%%%%%%%%%%%%%%%%%%%%%%%%%%%%%%%%%%%%%%%

\subsection{Progress Overlap Lemmas}
\label{Sec:samplLemCol}

In this section, we prove several simple lemmas that upper bound the progress overlap when making one classical or quantum query. Roughly speaking, these quantities correspond to the probability of recording new collisions in the history-database register when a new coordinate of the input is revealed by a query.

We first give a central fact that will be used throughout the next sections. It describes certain subspaces that remain orthogonal after applying one (classical or quantum) query to them.

\begin{fact}
  \label{Fact:ortho}
  The following linear maps are equal to zero over the subspace $\alg$ of consistent states:
    \[\pnc \qq \pc,\ \pc \qq \pnc,\ \pqh \qq \pnqnh,\ \pnh \qq \pnxh\]
  and
    \[\pnc \cc \pc,\ \pq \cc \pnq,\ \pnq \cc \pnxq,\ \pnh \cc \pnxh\, .\]
  For any states $\ket{\phi_1},\ket{\phi_2} \in \alg$ and basis-state predicate $\predP$, the following vectors are orthogonal:
    \[\pp \qc\pnxq \ket{\phi_1} \perp \qc(\pnq+\pxq) \ket{\phi_2} \quad \text{and} \quad
      \pp  \qc\pnxh \ket{\phi_1} \perp \qc(\pnh+\pxh) \ket{\phi_2}.\]
 for $b \in \rn$.
\end{fact}

\begin{proof}
  The statement follows by simple applications of \Cref{Lem:recordQ,Lem:recordC}.

  We detail the proof of the equality $\pq \cc \pnq = 0$. Consider any basis state $\ket{x,p,w,H,D} \in \supp{\pnq}$. By \Cref{Lem:recordC}, every history-database $(H',D')$ contained in the support of the post-query state $\cc \ket{x,p,w,H,D}$ must be identical to $(H,D)$ except possibly on the value $x$. Furthermore, since~$x$ must be in the history after the classical query (i.e. $H'(x) \neq \star$) it cannot contribute to any quantum collision in $(H',D')$. Thus, no quantum collision can be contained in~$(H',D')$.

  We sketch the proof of $\pp \cc \pnxh \ket{\phi_1} \perp \cc \pnh \ket{\phi_2}$. Every basis state in the support of $\cc \pnxh \ket{\phi_1}$ has a hybrid collision that does not contain the query index. On the other hand, none of the basis states in the support of $\cc \pnh \ket{\phi_2}$ satisfy this property since the only possible hybrid collisions must contain the index on which $\cc$ is queried. Hence, $\cc \pnxh \ket{\phi_1} \perp \cc \pnh \ket{\phi_2}$. Finally, applying $\pp$ does not change the orthogonality property since it can only remove basis states from the support of these states.
\end{proof}

We now analyze the effect of quantum and classical queries on the progress overlaps $\Gamma_0(\Pi,\kp), \Gamma_1(\Pi,\kp) \in [0,1]$ for different projectors $\Pi$. Recall that, by \Cref{Def:progIncr}, these numbers give the relative amplitude that moves from the support of $\id - \Pi$ to the support of~$\Pi$ after making a query, i.e. $\Gamma_b(\Pi,\kp) = \norm{\Pi \qc (\id-\Pi) \kp}^2 / \norm{(\id - \Pi) \kp}^2$. Notice that \Cref{Fact:ortho} already shows that $\Gamma_0(\pc,\kp) = \Gamma_0(\pnc,\kp) = \Gamma_1(\pnc,\kp) = 0$.

\begin{lemma}
  \label{Lem:sampl}
  Given an integer $t$ and a state $\kp \in \alg$, the progress overlap caused by one quantum query on $\kp$ are at most,
  \begin{align*}
    \Gamma_0(\pq,\kp) & \leq \frac{10 t}{N}\, , \Label{Eq:pqQ} & \Gamma_0(\pqorh,\kp) & \leq \frac{10t}{N}\, , \Label{Eq:pqorhQ} \\
    \Gamma_0(\ph,\kp) & \leq \frac{10 t}{N}\, , \Label{Eq:phQ} & & &&
  \end{align*}
  and the progress overlap caused by one classical query on $\kp$ are,
  \begin{align*}
    \Gamma_1(\pq,\kp) & = 0\, , \Label{Eq:pqC} & \Gamma_1(\pqorh,\kp) & \leq \frac{2t}{N}\, . \Label{Eq:pqorhC}
  \end{align*}
\end{lemma}

\begin{proof}
  The inequalities for quantum queries follow from \Cref{lem:quantum_query_progress} as $\predQ, \predQ+\predH$ and $\predH$ are history-database predicates (\Cref{Def:hdPred}) with the $\gamma$ parameters being $t/N$. Similarly, for classical queries, the two inequalities follow from \Cref{lem:classical_query_progress} with the $\gamma$ parameters being~$0$ and~$t/N$ respectively.
\end{proof}

Finally, we give four inequalities that do not follow from \Cref{lem:quantum_query_progress,lem:classical_query_progress}. \Cref{Eq:pnhorcQ,Eq:pnhorcC} below upper bound the progress made towards \emph{removing} all hybrid and classical collisions from the history-database, which is not a database monotone property (see~\Cref{Def:hdPred}). The purpose of \Cref{Eq:pxhC} is to upper bound the probability that a classical query transfers the query index $x$ from one hybrid collision to a \emph{different} hybrid collision. Finally, \Cref{Eq:phorcC} overcomes the fact that the predicate $\predH+\predC$ does not satisfy the condition stated in~\Cref{Eq:predC2}.

\begin{lemma}
  \label{Lem:samplH}
  Given an integer $t$ and a state $\kp \in \alg$, we have
  \begin{align*}
    \Gamma_0(\pnhorc,\kp) & \leq \frac{10t}{N}\, , \Label{Eq:pnhorcQ} & \norm{\ph \cc \pxh \kp}^2 & \leq \frac{t}{N} \cdot \norm{\pxh \kp}^2\, , \Label{Eq:pxhC} \\
    \Gamma_1(\pnhorc,\kp) & \leq \frac{2t}{N}\, , \Label{Eq:pnhorcC} & \norm{\pc \cc \pnhorc \kp}^2 & \leq \frac{2 t}{N} \cdot \norm{\pnhorc \kp}^2\, . \Label{Eq:phorcC}
  \end{align*}
\end{lemma}

The proofs of these equations use similar ideas to those of \Cref{lem:quantum_query_progress,lem:classical_query_progress}. They are deferred to \Cref{app:collision}.

%%%%%%%%%%%%%%%%%%%%%%%%%%%%%%%%%%%%%%%%%%%%%%%%%%%%%%%%%%%%%%%%%%%%%%%%%%%%%%%%

\subsection{Progress Increase Lemmas}
\label{Sec:Ccoll}

In this section, we analyze the progress measures for: (1) finding a classical collision, (2) finding a hybrid collision but no classical ones and (3) finding quantum collisions only. We start with the case of quantum queries.

\begin{rlemma}[\Cref{Lem:quProg} \normalfont{(Progress Measure, Quantum Query)}]
  Given an integer $t$ and a state $\kp \in \alg$ with norm at most 1, the progress caused by one quantum query on $\kp$ satisfies
    \begin{align*}
        \Delta_0(\pc,\kp) & = 0 \, , \Label{Eq:DpcQ} \\
        \Delta_0(\phnc,\kp) & \leq 2\sqrt{\frac{10t}{N}}\norm{\pxhnc \kp} + \frac{10t}{N} \, , \Label{Eq:DphncQ} \\
        \Delta_0(\pqnhnc,\kp) & \leq \sqrt{\frac{8t}{N}} \norm{\pxhnc\kp} + 2\sqrt{\frac{20t}{N}} \norm{\pxqnhnc \kp} + \frac{20t}{N} \, . \Label{Eq:DpqnhncQ}
    \end{align*}
\end{rlemma}

\begin{proof}[Proof of \Cref{Eq:DpcQ}]
  We have $\norm{\pc \qq \kp}^2 = \norm{\pc \qq \pc\kp}^2 = \norm{\pc\kp}^2$ since $\id = \pc + \pnc$ and $\pc \qq \pnc = \pnc \qq \pc = 0$ by \Cref{Fact:ortho}.
\end{proof}

\begin{proof}[Proof of \Cref{Eq:DphncQ}]
  We use the decomposition $\id = \pnxhnc + \pxhnc + \pnhnc + \pc$. By \Cref{Fact:ortho}, $\phnc \qq \pc = 0$ and the states $\phnc \qq \pnxhnc\kp = \phnc \qq \pnxh \pnc\kp$ and $\phnc \qq (\pxhnc + \pnhnc)\kp = \phnc \qq (\pxh + \pnh)\pnc\kp$ are orthogonal. Therefore,
  \begin{align*}
      \norm{\phnc \qq \kp}^2
        & = \norm{\phnc \qq (\pnxhnc + \pxhnc + \pnhnc + \pc)\kp}^2 \\
        & = \norm{\phnc \qq \pnxhnc\kp}^2 + \norm{\phnc \qq (\pxhnc + \pnhnc)\kp}^2 \\
        & \leq \norm{\pnxhnc\kp}^2 + \pt*{\norm{\pxhnc\kp} + \norm{\ph \qq \pnh\kp}}^2 \\
        & = \norm{\phnc \kp}^2 + 2 \norm{\pxhnc\kp} \cdot \norm{\ph \qq \pnh\kp} + \norm{\ph \qq \pnh\kp}^2
  \end{align*}
  where the third line uses the triangle inequality, and the last line uses that $\norm{\phnc \kp}^2 = \norm{\pnxhnc\kp}^2 + \norm{\pxhnc\kp}^2$. Finally, $\norm{\ph \qq \pnh\kp}^2 \leq 10t/N$ by \Cref{Eq:phQ}.
\end{proof}

\begin{proof}[Proof of \Cref{Eq:DpqnhncQ}]
  We use the decomposition $\id = \pnqnhnc + \pnxqnhnc + \pxqnhnc + \pxhnc + \pnxhnc + \pc$. By \Cref{Fact:ortho}, $\pqnhnc \qq (\pnxhnc + \pc) = 0$ and the states $\pqnhnc \qq \pnxqnhnc\kp$ and $\pqnhnc \qq (\pnqnhnc + \pxqnhnc)\kp$ are orthogonal. Therefore,
  \begin{align*}
      & \norm{\pqnhnc \qq \kp}^2 \\
        & \quad = \norm{\pqnhnc \qq (\pnqnhnc + \pnxqnhnc + \pxqnhnc + \pxhnc + \pnxhnc + \pc)\kp}^2 \\
        & \quad \leq \norm{\pqnhnc \qq \pnxqnhnc\kp}^2 + \norm{\pqnhnc \qq (\pnqnhnc + \pxqnhnc)\kp}^2 + 3\norm{\pqnhnc \qq \pxhnc\kp} \\
        & \quad \leq \norm{\pnxqnhnc\kp}^2 + \pt*{\norm{\pq \qq \pnq\kp} + \norm{\pxqnhnc\kp}}^2 + 3\norm{\pnh \qq \pxhnc\kp} \\
        & \quad = \norm{\pqnhnc \kp}^2 + 2 \norm{\pxqnhnc\kp} \cdot \norm{\pq \qq \pnq\kp} + \norm{\pq \qq \pnq\kp}^2 + 3\norm{\pnh \qq \pxhnc\kp}
  \end{align*}
  where the second line uses the identity $\norm{a+b}^2 \leq \norm{a}^2 + 3\norm{b}$ when $\norm{a},\norm{b} \leq 1$, the third line uses the triangle inequality and the last line uses that $\norm{\pqnhnc \kp}^2 = \norm{\pnxqnhnc\kp}^2 + \norm{\pxqnhnc\kp}$. Finally, $\norm{\pq \qq \pnq\kp}^2 \leq 10t/N$ by \Cref{Eq:pqQ} and $\norm{\pnh \qq \pxhnc\kp}^2 \leq (10t/N) \norm{\pxhnc\kp}^2$ by \Cref{Eq:pnhorcQ}.
\end{proof}

We now analyze the case of classical queries.

\begin{rlemma}[\Cref{Lem:clProg} \normalfont{(Progress Measure, Classical Query)}]
  Given an integer $t$ and a state $\kp \in \alg$ with norm at most 1, the progress caused by one classical query on $\kp$ are at most,
    \begin{align*}
      \Delta_1(\pc,\kp) & \leq 2 \norm{\pxhnc \kp}^2 + \delta_1 + \frac{4 t}{N} \, , \Label{Eq:DpcC} \\
      \Delta_1(\phnc,\kp) & \leq - \norm{\pxhnc \kp}^2 + 2\norm{\pxqnhnc \kp}^2 - \delta_1 + 2\delta_2 + \frac{12t}{N} \, ,\Label{Eq:DphncC} \\
      \Delta_1(\pqnhnc,\kp) & \leq \sqrt{\frac{2t}{N}} \norm{\pxhnc\kp}  - \norm{\pxqnhnc \kp}^2 - \delta_2 \, \Label{Eq:DpqnhncC}
    \end{align*}
  where $\delta_1 = \norm{\pc \cc \pnxhnc \kp}^2$ and $\delta_2 = \norm{\phnc \cc \pnxqnhnc \kp}^2$.
\end{rlemma}

\begin{proof}[Proof of \Cref{Eq:DpcC}]
  We use the decomposition $\id = \pc + \pxhnc + \pnxhnc + \pnhnc$. By \Cref{Fact:ortho}, the states $\pc \cc \pc\kp$, $\pc \cc \pnxhnc\kp$ and $\pc \cc (\pxhnc + \pnhnc)\kp$ are orthogonal. Therefore,
  \begin{align*}
      \norm{\pc \cc \kp}^2
        & = \norm{\pc \cc (\pc + \pxhnc + \pnxhnc + \pnhnc)}^2 \\
        & = \norm{\pc \cc \pc}^2 + \norm{\pc \cc \pnxhnc}^2 + \norm{\pc \cc (\pxhnc + \pnhnc)\kp}^2 \\
        & \leq \norm{\pc}^2 + \norm{\pc \cc \pnxhnc}^2 + 2 \norm{\pxhnc \kp}^2 + 2\norm{\pc \cc \pnhnc\kp}^2
  \end{align*}
  where the last line uses the identity $\norm{a+b}^2 \leq 2\norm{a}^2 + 2\norm{b}^2$. Finally, $\norm{\pc \cc \pnhnc\kp}^2 \leq 2t/N$ by \Cref{Eq:phorcC}.
\end{proof}

\begin{proof}[Proof of \Cref{Eq:DphncC}]
  We use the decomposition $\id = \pc + \pnxhnc + \pxhnc + \pnhnc$. By \Cref{Fact:ortho}, $\phnc \cc \pc = 0$ and the states $\phnc \cc \pnxhnc\kp$ and $\phnc \cc (\pxhnc + \pnhnc)\kp$ are orthogonal. Therefore,
  \begin{align*}
      \norm{\phnc \cc \kp}^2
        & = \norm{\phnc \cc \pnxhnc}^2 + \norm{\phnc \cc (\pxhnc + \pnhnc)\kp}^2 \\
        & = \norm{\phnc}^2 - \norm{\pxhnc}^2 - \norm{\pc \cc \pnxhnc}^2 + \norm{\phnc \cc (\pxhnc + \pnhnc)\kp}^2
  \end{align*}
  where the second line uses that $\norm{\phnc}^2 - \norm{\pxhnc}^2 - \norm{\pc \cc \pnxhnc}^2 = \norm{\pnxhnc}^2 - \norm{\pc \cc \pnxhnc}^2 = \norm{\pnc \cc \pnxhnc}^2 = \norm{\phnc \cc \pnxhnc}^2$ since $\pnhnc \cc \pnxhnc = 0$ by \Cref{Fact:ortho}. It remains to bound $\norm{\phnc \cc (\pxhnc + \pnhnc)\kp}^2$. We further decompose $\pnhnc$ into $\pnhnc = \pxqnhnc + \pnxqnhnc + \pnqnhnc$ and observe that $\phnc \cc \pxqnhnc\kp$ and $\phnc \cc \pnxqnhnc\kp$ are orthogonal by \Cref{Fact:ortho}. Hence,
  \begin{align*}
    & \norm{\phnc \cc (\pxhnc + \pnhnc)\kp}^2 \\
      & \qquad \leq 2 \norm{\phnc \cc (\pxqnhnc + \pnxqnhnc) \kp}^2 + 2 \norm{\phnc \cc (\pxhnc + \pnqnhnc)\kp}^2 \\
      & \qquad = 2 \norm{\phnc \cc \pxqnhnc \kp}^2 + 2 \norm{\phnc \cc \pnxqnhnc \kp}^2 + 2\norm{\phnc \cc (\pxhnc + \pnqnhnc)\kp}^2 \\
      & \qquad \leq 2 \norm{\pxqnhnc \kp}^2 + 2 \norm{\phnc \cc \pnxqnhnc \kp}^2 + \frac{12t}{N}
  \end{align*}
  where the last line uses the triangle inequality and \Cref{Eq:pxhC,Eq:pqorhC} on $\norm{\phnc \cc (\pxhnc + \pnqnhnc)\kp}^2 \leq (\norm{\phnc \cc \pxhnc\kp} + \norm{\phnc \cc\pnqnhnc)\kp})^2 \leq (1+\sqrt{2})^2 t/N$.
\end{proof}

\begin{proof}[Proof of \Cref{Eq:DpqnhncC}]
  We use the decomposition $\id = \pnqnhnc+\pnxqnhnc+\pxqnhnc+\pnxhnc+\pxhnc+\pc$. By \Cref{Fact:ortho}, $\pqnhnc \cc (\pnqnhnc+\pxqnhnc+\pnxhnc+\pc) = 0$. Therefore,
  \begin{align*}
      \norm{\pqnhnc \cc \kp}^2
        & = \norm{\pqnhnc \cc (\pnxqnhnc+\pxhnc) \kp}^2 \\
        & \leq \norm{\pqnhnc \cc \pnxqnhnc\kp}^2 + 3\norm{\pqnhnc \cc \pxhnc\kp} \\
        & \leq \norm{\pqnhnc \kp}^2 - \norm{\pxqnhnc\kp}^2 - \norm{\phnc \cc \pnxqnhnc\kp}^2 + 3\norm{\pqnhnc \cc \pxhnc\kp}
  \end{align*}
  where the last line uses $\norm{\pqnhnc \cc \pnxqnhnc\kp}^2 \leq \norm{(\pnhnc +\pc)\cc \pnxqnhnc\kp}^2 = \norm{\pqnhnc \kp}^2 - \norm{\pxqnhnc\kp}^2 - \norm{\phnc \cc \pnxqnhnc\kp}^2$. Finally, $\norm{\pqnhnc \cc \pxhnc\kp}^2 \leq (2t/N)\norm{\pxhnc\kp}^2$ by \Cref{Eq:pnhorcC}.
\end{proof}

\section*{Acknowledgements}
The authors would like to thank Ansis Rosmanis for fruitful discussions and for sharing a draft of his work on noisy oracles~\cite{Ros23p}. The authors are also grateful to the anonymous referees for their valuable comments and suggestions which helped to improve the paper. Part of this work was supported by the Simons Institute through Simons-Berkeley Postdoctoral Fellowships.

%#############################################################################################################%

\printbibliography[heading=bibintoc]

\appendix

\section{\texorpdfstring{Missing Proofs for Section \ref{sec:compressed_oracle_for_hybrid} (Hybrid Compressed Oracle)}{Missing Proofs for Hybrid Compressed Oracle}}

\subsection{\texorpdfstring{Resampling Lemma (\Cref{Lem:recordQ,Lem:recordC})}{Resampling Lemma}}
\label{sec:proof_resample}

\begin{proof}[Proof of \Cref{Lem:recordQ}]
  We only prove the third item, corresponding to $H(x) = \star$, $D(x) \neq \bot$, $p \neq 0$, as it is the most involved of the three. The operator $\qq = S \oq S$ appends $\ket{0}$ to the workspace register and acts as a control on all registers except $\CD_x$, which contains the $x$-th entry of the database. Let $z = D(x) \neq \bot$ be the value in $\CD_x$. Writing $\ket{z} = \frac{1}{\sqrt{N}}\sum_{p \in [N]} \omega^{-pz} \ket{\phat}$ in the Fourier basis, we can see that $S$ maps $\ket{z}_{\CD_x}$ to
    \[\frac{1}{\sqrt{N}}\sum_{p \in [N]} \omega^{-pz} \ket{\phat} + \frac1{\sqrt{N}}\ket{\bot} - \frac1{\sqrt{N}}\ket{\hat{0}} = \ket{z} + \frac1{\sqrt{N}}\ket{\bot} - \frac1{N}\sum_{y \in [N]}\ket{y}.\]
  Applying $\oq$ to the above state, we get
    \[ \omega^{pz}\ket{z} + \frac{1}{\sqrt{N}}\ket{\bot} - \frac1{N}\sum_{y \in [N]} \omega^{py} \ket{y} = \frac{\omega^{pz}}{\sqrt{N}} \sum_{p' \in [N]} \omega^{-p'z} \ket{\hat{p}'} + \frac{1}{\sqrt{N}}\ket{\bot} - \frac{1}{\sqrt{N}} \ket{\hat{p}}.\]
  Applying $S$ again to the above and simplifying we get
    \[\frac{\omega^{pz}}{\sqrt{N}} \sum_{p' \in [N]} \omega^{-p'z} \ket{\hat{p}'} + \frac{\omega^{p z}}{\sqrt{N}} \ket{\bot} + \frac{1-\omega^{p z}}{\sqrt{N}}\ket{\hat{0}} - \frac{1}{\sqrt{N}} \ket{\hat{p}} = \omega^{p z} \ket{z} + \frac{\omega^{p z}}{\sqrt{N}} \ket{\bot} + \sum_{y \in [N]} \frac{1-\omega^{p z} - \omega^{p y}}{N} \ket{y}.\]
  thus proving the third item.
\end{proof}

\begin{proof}[Proof of \Cref{Lem:recordC}]
  We only prove the third item, corresponding to $H(x) = \star$, $D(x) \neq \bot$. Let $\ket{H}_{\mc{H}} = \ket{(x_1,y_1),\dots,(x_c,y_c),\star,\dots,\star}_{\mc{H}}$, for some integer $c$, denote the value contained in the history register. The operator $\cc = S \oc S$ appends $\ket{1}$ to the workspace register and acts as a control on all registers except $\CH_{c+1}\CD_x$, which contain $\ket{\star,z}_{\CH_{c+1}\CD_x}$ for some $z = D(x) \neq \bot$. Similarly as in the above proof of \Cref{Lem:recordQ}, after applying the first two operators $\oc S$, this state gets mapped to
    \[\omega^{pz}\ket{(x,z),z} + \frac{1}{\sqrt{N}}\ket{(x,\bot),\bot} - \frac1{N}\sum_{y \in [N]} \omega^{py} \ket{(x,y),y}.\]
  where the value contained in the database register $\CD_x$ has been appended to the history (by definition of a the classical query operator $\oc$). Finally, applying $S$ to the above state does nothing since the query index $x$ is now contained in the history.
\end{proof}

%%%%%%%%%%%%%%%%%%%%%%%%%%%%%%%%%%%%%%%%%%%%%%%%%%%%%%%%%%
%%%%%%%%%%%%%%%%%%%%%%%%%%%%%%%%%%%%%%%%%%%%%%%%%%%%%%%%%%

\subsection{\texorpdfstring{Progress Overlap Lemmas (\Cref{lem:quantum_query_progress,lem:classical_query_progress})}{Progress Overlap Lemmas}}
\label{sec:proof_progress_measure}

We first give the proof for \Cref{lem:classical_query_progress} (classical query) as it differs the most from previous work on the compressed oracle. The proof will be next adapted for \Cref{lem:quantum_query_progress} (quantum query).

\begin{proof}[Proof of \Cref{lem:classical_query_progress}]
  Let $\pnp \ket{\phi} = \sum_{x,p,w,H,D} \alpha_{x, p, w, H,D} \ket{x, p, w}\ket{H,D} \in \alg \cap \supp{\pnp}$ be any state supported over consistent basis-states evaluating the predicate $\predP$ to false. We show that, after making a classical query, the probability of satisfying $\predP$ is at most $\norm{\pp \cc \pnp \kp}^2 \leq 2\gamma \cdot \norm{\pnp \kp}^2$. We define three projections $\Pi_1, \Pi_2, \Pi_3$ such that $\Pi_1 + \Pi_2 + \Pi_3 = \Pi_{\br{\predP}}$.
  \begin{itemize}
      \item $\Pi_1$: all basis states $\ket{x, p, w, H,D} \in \alg \cap \supp{\pnp}$ such that $H(x) = \star$ and $D(x) = \bot$.
      \item $\Pi_2$: all basis states $\ket{x, p, w, H,D} \in \alg \cap \supp{\pnp}$ such that $H(x) = \star$ and $D(x) \neq \bot$.
      \item $\Pi_3$: all basis states $\ket{x, p, w, H,D} \in \alg \cap \supp{\pnp}$ such that $H(x) \neq \star$.
  \end{itemize}
  Below, we prove the inequalities $\norm{\pp \cc \Pi_1 \kp}^2 \leq \gamma \norm{\Pi_1 \kp}^2$, $\norm{\pp \cc \Pi_2 \kp}^2 \leq \gamma \norm{\Pi_2 \kp}^2$ and $\norm{\pp \cc \Pi_3 \kp} = 0$. Hence, by the triangle inequality and Cauchy--Schwarz inequality, we conclude that
  \[
    \norm{\pp \cc \Pi_{\br{\predP}} \kp}^2 \leq \pt*{\norm{\pp \cc \Pi_1 \kp} + \norm{\pp \cc \Pi_2 \kp} + \norm{\pp \cc \Pi_3 \kp}}^2 \leq 2\gamma \norm{\pnp \kp}^2.
  \]

  \paragraph{Analysis of $\Pi_1$.} The projection $\Pi_1$ corresponds to \emph{sampling} a new outcome at $x$. We have
  \begin{align*}
      \norm{\pp \cc \Pi_1 \kp}^2
      & = \norm[\Bigg]{\pp  \cc \sum_{\substack{x, p, w,H,D : \\ H(x) = \star,D(x) = \bot}} \alpha_{x, p, w, H,D} \ket{x, p, w}\ket{H,D}}^2 \\
      & = \norm[\Bigg]{\pp \sum_{\substack{x, p, w,H,D : \\ H(x) = \star,D(x) = \bot}}  \alpha_{x, p, w, H,D} \ket{x, p, w1} \pt[\bigg]{\sum_{y \in [N]} \frac{\omega^{p y}}{\sqrt{N}} \ket{H_{x \la y},D_{x \la y}}}}^2 \\
      & = \sum_{\substack{x, p, w,H,D : \\ H(x) = \star,D(x) = \bot}}  \abs{\alpha_{x, p, w, H,D}}^2 \cdot \Pr_{y \gets [N]}\bc*{\pt*{H_{x \la y},D_{x \la y}} \in \predP^{-1}(\predT)} \\
      & \leq \gamma \norm{\Pi_1 \kp}^2.
  \end{align*}
  The first line is by definition of $\Pi_1$. The second line is by \Cref{Lem:recordC}. The third line uses the orthogonality of the basis states. Finally, the last line is by Equation~(\ref{Eq:predC1}).

  \paragraph{Analysis of $\Pi_2$.} The projection $\Pi_2$ corresponds to \emph{resampling} a new outcome at index $x$ (see the third item of \Cref{Lem:recordC}). There are three components and the only states that may be in the support of $\pp$ after the query is done are those for which $D(x)$ is resampled to a different value $y \neq D(x)$. Indeed, the other two cases are where $D(x) = \bot$ gets removed or $D(x)$ remains unchanged in the database. The former case cannot make the predicate true because of the database monotone property (Definition~\ref{Def:hdPred}), the latter case cannot either because of the condition stated in Equation~(\ref{Eq:predC2}). Hence, we have
  \begin{align*}
    \norm{\pp \cc \Pi_2 \kp}^2
      & = \norm[\bigg]{\pp \cc \sum_{\substack{x, p, w,H,D : \\ H(x) = \star, D(x) \neq \bot}} \alpha_{x, p, w, H,D} \ket{x, p, w}\ket{H,D}}^2  \\
      & = \norm[\bigg]{ \pp \sum_{\substack{x, p, w,H,D : \\ H(x) = \star, D(x) \neq \bot}} \sum_{y \in [N]} \alpha_{x, p, w, H,D} \frac{\omega^{p y}}{N} \ket{x, p, w1}\ket{H_{x \la y},D_{x \la y}}}^2.
  \end{align*}

  Next, observe that for any two tuples $(x,p,w,H,D_{x \la \bot},y) \neq (x',p',w',H',D'_{x' \la \bot},y')$, the basis states $\ket{x, p, w1}\ket{H_{x \la y},D_{x \la y}}$ and $\ket{x', p', w'1}\ket{H'_{x' \la y'},D'_{x' \la y'}}$ must be orthogonal. Thus, we can exploit this orthogonality property to simplify the above expression as follows.
  \begin{align*}
    \norm{\pp \cc \Pi_2 \kp}^2
      & = \sum_{\substack{x, p, w, H,D, y : \\ y \in [N], H(x) = \star, D(x) = y}} \norm[\bigg]{\pp \sum_{z \in [N]} \alpha_{x, p, w, H,D_{x \la z}} \frac{\omega^{p y}}{N} \ket{x, p, w1}\ket{H_{x \la y},D}}^2  \\
      & = \sum_{\substack{x, p, w, H,D, y : \\ y \in [N], H(x) = \star, D(x) = y, \\ \predP(H_{x \la y},D) = \predT}} \abs[\bigg]{\sum_{z \in [N]} \alpha_{x, p, w, H,D_{x \la z}} \frac{\omega^{p y}}{N}}^2.
  \end{align*}

  Applying the Cauchy--Schwarz inequality, we have
  \begin{align*}
    \norm{\pp \cc \Pi_2 \kp}^2
      & \leq \sum_{\substack{x, p, w, H,D, y : \\ y \in [N], H(x) = \star, D(x) = y, \\ \predP(H_{x \la y},D) = \predT}} \sum_{z \in [N]} \frac{\abs{\alpha_{x, p, w, H,D_{x \la z}}}^2}{N} \\
      & = \sum_{\substack{x, p, w, H,D : \\ H(x) = \star, D(x) \neq \bot}} \pt[\Bigg]{\sum_{y \in [N]:  \predP(H_{x \la y},D_{x \la y}) = \predT} \frac{\abs{\alpha_{x, p, w, H,D}}^2}{N}} \\
      & = \sum_{\substack{x, p, w, H,D : \\ H(x) = \star, D(x) \neq \bot}} \abs{\alpha_{x, p, w, H,D}}^2 \cdot \Pr_{y \gets [N]}\bc*{(H_{x \la y},D_{x \la y}) \in \predP^{-1}(\predT)}.
  \end{align*}

  Finally, for each $\ket{x,p,w,H,D}$ in the support of $\Pi_2$, we must have $\predP(H,D_{x \la \bot}) = \predF$ by the database monotone property (see Definition~\ref{Def:hdPred}). Hence, by Equation~(\ref{Eq:predC1}), the above inequality implies that $\norm{\pp \cc \Pi_2 \kp}^2 \leq \gamma \cdot \norm{\Pi_2 \kp}^2$.

  \paragraph{Analysis of $\Pi_3$.} By \Cref{Lem:recordC}, the operator $\cc$ maps any state $\ket{x, p, w}\ket{H,D} \in \supp{\Pi_3}$ to $\omega^{p D(x)} \ket{x, p, w1}\ket{H_{x \la D(x)},D}$ since $H(x) \neq \star$.  Moreover, $H$ and $H_{x \la D(x)}$ have the same function representation (since the initial state is history-database consistent). Thus, by the history invariant property (see Definition~\ref{Def:hdPred}), we have $\predP(H_{x \la D(x)},D) = \predF$ and $\norm{\pp \cc \Pi_3 \kp} = 0$.
\end{proof}

The proof of \Cref{lem:quantum_query_progress} is similar to the above one, the main difference being that quantum queries do not act on the history register.

\begin{proof}[Proof of \Cref{lem:quantum_query_progress}]
  Let $\pnp \ket{\phi} = \sum_{x,p,w,H,D} \alpha_{x, p, w, H,D} \ket{x, p, w}\ket{H,D} \in \alg \cap \supp{\pnp}$. We will prove that $\norm{\pp \qq \pnp \kp}^2 \leq 10\gamma \cdot \norm{\pnp \kp}^2$. We first define three projections $\Pi_1, \Pi_2, \Pi_3$ such that $\Pi_1 + \Pi_2 + \Pi_3 = \Pi_{\br{\predP}}$.
  \begin{itemize}
      \item $\Pi_1$: all basis states $\ket{x, p, w, H,D} \in \alg \cap \supp{\pnp}$ such that $H(x) = \star$, $D(x) = \bot$, $p \neq 0$.
      \item $\Pi_2$: all basis states $\ket{x, p, w, H,D} \in \alg \cap \supp{\pnp}$ such that $H(x) = \star$, $D(x) \neq \bot$, $p \neq 0$.
      \item $\Pi_3$: all basis states $\ket{x, p, w, H,D} \in \alg \cap \supp{\pnp}$ such that $H(x) \neq \star$ or $p = 0$.
  \end{itemize}

  Below, we prove that $\norm{\pp \qq \Pi_1 \kp}^2 \leq \gamma \norm{\Pi_1 \kp}^2$, $\norm{\pp \qq \Pi_2 \kp}^2 \leq 9\gamma \norm{\Pi_2 \kp}^2$ and $\norm{\pp \qq \Pi_3 \kp} = 0$. Hence, by the triangle and Cauchy--Schwarz inequalities, we conclude that $\norm{\pp \qq \pnp \kp}^2 \leq 10\gamma \cdot \norm{\pnp \kp}^2$.

  \paragraph{Analysis of $\Pi_1$.} The effect of applying $\qq$ on a basis state in the support of $\Pi_1$ is described in the second item of \Cref{Lem:recordQ}. Similarly to the analysis of $\Pi_1$ in the proof of \Cref{lem:classical_query_progress}, we deduce that
  \begin{align*}
      \norm{\pp \qq \Pi_1 \kp}^2
      & = \sum_{\substack{x, p, w,H,D : \\ H(x) = \star,D(x) = \bot,p \neq 0}}  \abs{\alpha_{x, p, w, H,D}}^2 \cdot \Pr_{y \gets [N]}\bc*{\pt*{H,D_{x \la y}} \in \predP^{-1}(\predT)} \\
      & \leq \gamma \norm{\Pi_1 \kp}^2.
  \end{align*}
  where the second line is by Equation~(\ref{Eq:predQ}).

  \paragraph{Analysis of $\Pi_2$.} The effect of applying $\qq$ on a basis state in the support of $\Pi_2$ is described in the third item of \Cref{Lem:recordQ}. By using the bound $\abs{1-\omega^{p D(x)} - \omega^{p y}} \leq 3$ on the term displayed there, we can follow a similar analysis as in the proof of \Cref{lem:classical_query_progress} for $\Pi_2$ and deduce that
  \begin{align*}
    \norm{\pp \qq \Pi_2 \kp}^2
      & \leq 9 \sum_{\substack{x, p, w, H,D : \\ H(x) = \star, D(x) \neq \bot, p \neq 0}} \abs{\alpha_{x, p, w, H,D}}^2 \cdot \Pr_{y \gets [N]}\bc*{(H,D_{x \la y}) \in \predP^{-1}(\predT)} \\
      & \leq 9 \gamma \norm{\Pi_2 \kp}^2.
  \end{align*}
  where the second line is by Equation~(\ref{Eq:predQ}).

  \paragraph{Analysis of $\Pi_3$.} By the first item in \Cref{Lem:recordQ}, the operator $\qq$ maps any basis state in the support of $\Pi_3$ to itself, up to a phase factor. Thus, we have $\norm{\pp \qq \Pi_3 \kp} = 0$.
\end{proof}

%\yassine{$p \neq 0$ and $\sum_y \abs{1-\omega^{p D(x)} - \omega^{p y}}^2 = N \pt*{\abs{1-\omega^{p D(x)}}^2 + 1} = N \pt*{4 \sin^2\pt*{\frac{\pi p D(x)}{N}} + 1} \leq 5N$}

\section{\texorpdfstring{Missing Proofs for Section \ref{sec:hybrid_cols} (Collision Finding)}{Missing Proofs for Collision Finding}}
\label{app:collision}
In this section, we prove the following lemma:

\begin{rlemma}[\Cref{Lem:samplH} {\normalfont (Restated)}]
  Given an integer $t$ and a state $\kp \in \alg$, we have
  \begin{align*}
    \Gamma_0(\pnhorc,\kp) & \leq \frac{10t}{N}\, , && (\ref{Eq:pnhorcQ}) & \norm{\ph \cc \pxh \kp}^2 & \leq \frac{t}{N} \cdot \norm{\pxh \kp}^2\, , && (\ref{Eq:pxhC}) \\
    \Gamma_1(\pnhorc,\kp) & \leq \frac{2t}{N}\, , && (\ref{Eq:pnhorcC}) & \norm{\pc \cc \pnhorc \kp}^2 & \leq \frac{2 t}{N} \cdot \norm{\pnhorc \kp}^2\, . && (\ref{Eq:phorcC})
  \end{align*}
\end{rlemma}

We will use the following simple fact about the predicate $\predX\predH$.

\begin{fact}
  \label{Fact:XH}
  For any basis state $\ket{x,p,w,H,B,D}$ satisfying the predicate $\predX\predH$, we have
    \begin{enumerate}
      \item The query index $x$ is in the database but not in the history, that is $D(x) \neq \bot$ and $H(x) = \star$.
      \item There is no hybrid collision in $(H,D_{x \la \bot})$.
      \item The query index $x$ does not belong to a quantum collision.
    \end{enumerate}
\end{fact}

\begin{proof}
  The first two items are immediate by definition of $\predX$ and $\predH$. For the last item, if $x$ was in a quantum collision then, since it also belongs to a hybrid collision, there would exist a second hybrid collision that does not contain $x$ (which contradicts $\predX$).
\end{proof}

Since the proofs of \Cref{Eq:pnhorcQ,Eq:pxhC,Eq:pnhorcC,Eq:phorcC} share strong similarities with those of \Cref{Lem:recordQ,Lem:recordC}, we skip some details in the calculation below.

\begin{proof}[Proof of \Cref{Eq:pnhorcQ}]
  We first claim that it is sufficient to show that
  \begin{equation}
    \norm{\pnh \qq \pxh \kp}^2 \leq \frac{10t}{N} \cdot \norm{\pxh \kp}^2. \label{Eq:pxhQ}
  \end{equation}
  Indeed, $\pnhorc \qq \phorc \kp = \pnhorc \qq \pxh \phorc \kp$ by \Cref{Fact:ortho}. Thus, using \Cref{Eq:pxhQ}, we conclude that $\norm{\pnhorc \qq \phorc \kp}^2 \leq \norm{\pnh \qq \pxh \phorc \kp}^2 \leq \frac{10t}{N} \cdot \norm{\phorc \kp}^2$.

  We now prove \Cref{Eq:pxhQ}. Let $\pxh \ket{\phi} = \sum_{x,p,w,H,B,D} \alpha_{x, p, w, H,B,D} \ket{x, p, w,H,B,D} \in \alg \cap \supp{\pxh}$. Notice that if the phase register contains $p = 0$ then doing a quantum query on such a state will not modify $(H,D)$. Hence, we only need to consider the basis states for which~$p \neq 0$. Together with \Cref{Fact:XH}, it implies that the post-query state is given by the third item of \Cref{Lem:recordQ},
  \begin{align*}
     \pnh \qq \pxh \kp = \pnh \sum_{x,p,w,H,B,D} \alpha_{x, p, w, H, B,D} \ket{x,p,w,H,B_{\la 0}} & \biggl(\frac{\omega^{p D(x)}}{\sqrt{N}} \ket{D_{x \la \bot}} \\
     & \quad + \sum_{y \in [N]} \frac{1-\omega^{p D(x)} - \omega^{p y}}{N}  \ket{D_{x \la y}}\biggr).
  \end{align*}

  Next, using the orthogonality between basis states, the norm of the above state is equal to,
  \begin{align*}
    & \norm{\pnh \qq \pxh \kp}^2 \\
      & \quad = \sum_{\substack{x, p, w, H,B,D : \\ H(x) = \star, D(x) = \bot}} \norm[\bigg]{\pnh \sum_{z \in [N]} \alpha_{x, p, w, H,B,D_{x \la z}} \frac{\omega^{p z}}{\sqrt{N}} \ket{x, p, w}\ket{H,B_{\la 0},D}}^2  \\
      & \quad\quad + \sum_{\substack{x, p, w, H,B,D, y : \\ y \in [N], H(x) = \star, D(x) = y}} \norm[\bigg]{\pnh \sum_{z \in [N]} \alpha_{x, p, w, H,B,D_{x \la z}} \frac{1-\omega^{p z} - \omega^{p y}}{N} \ket{x, p, w}\ket{H,B_{\la 0},D}}^2 \\
      & \quad = \sum_{\substack{x, p, w, H,B,D : \\ H(x) = \star, D(x) = \bot, \\ \predH(x,p,w,H,B,D) = \predF}} \abs[\bigg]{\sum_{z \in [N]} \alpha_{x, p, w, H,B,D_{x \la z}} \frac{\omega^{pz}}{\sqrt{N}}}^2 \\
      & \quad\quad + \sum_{\substack{x, p, w, H,B,D, y : \\ y \in [N], H(x) = \star, D(x) = y, \\ \predH(x,p,w,H,B,D) = \predF}} \abs[\bigg]{\sum_{z \in [N]} \alpha_{x, p, w, H,B,D_{x \la z}} \frac{1-\omega^{p z} - \omega^{p y}}{N}}^2.
  \end{align*}

  By the Cauchy--Schwarz inequality, the above term is at most
  \begin{align*}
    & \norm{\pnh \qq \pxh \kp}^2 \\
    \ & \leq \sum_{\substack{x, p, w, H,B,D : \\ H(x) = \star, D(x) = \bot, \\ \predH(x,p,w,H,B,D) = \predF}} \pt[\bigg]{\sum_{z \in [N]} \abs{\alpha_{x, p, w, H,B,D_{x \la z}}}^2} \Pr_{z \gets [N]}\bc*{\predH(x,p,w,H,B,D_{x \la z}) = \predT} \\
      & + \sum_{\substack{x, p, w, H,B,D, y : \\ y \in [N], H(x) = \star, D(x) = y, \\ \predH(x,p,w,H,B,D_{x \la \bot}) = \predF}} \frac{9}{N} \pt[\bigg]{\sum_{z \in [N]} \abs{\alpha_{x, p, w, H,B,D_{x \la z}}}^2} \Pr_{z \gets [N]}\bc*{\predH(x,p,w,H,B,D_{x \la z}) = \predT} \\
      &  = \sum_{\substack{x, p, w, H,B,D : \\ H(x) = \star, D(x) = \bot, \\ \predH(x,p,w,H,B,D) = \predF}} 10 \pt[\bigg]{\sum_{z \in [N]} \abs{\alpha_{x, p, w, H,B,D_{x \la z}}}^2} \Pr_{z \gets [N]}\bc*{\predH(x,p,w,H,B,D_{x \la z}) = \predT}
  \end{align*}
  where we used that the non-zero amplitudes $\alpha_{x, p, w, H,B,D_{x \la z}}$ must satisfy $\predH(x,p,w,H,B,D_{x \la z}) = \predT$ (since $\pxh \ket{\phi} \in \supp{\ph}$), we extended the range of the second summation to all pairs $(H,D)$ that contain no hybrid collision in $(H,D_{x \la \bot})$ and we used that $\abs{1-\omega^{p z} - \omega^{p y}} \leq 3$.

  Finally, since $\pxh \kp$ is supported over basis states whose history register contains at most $t$ non-$\star$ entries, the probability to create a hybrid collision by adding one value to the database is at most $\Pr_{z \gets [N]}\bc*{\predH(x,p,w,H,B,D_{x \la z}) = \predT} \leq t/N$. We conclude that,
  $\norm{\pnh \qq \pxh \kp}^2 \leq \frac{10t}{N} \sum_{x,p,w,H,B,D} \abs{\alpha_{x, p, w, H,B,D}}^2 = \frac{10t}{N} \norm{\pxh \kp}^2$.
\end{proof}

\begin{proof}[Proof of \Cref{Eq:pnhorcC}]
  Similarly to the above proof, by \Cref{Fact:ortho}, it is sufficient to show that
    \begin{equation}
      \norm{\pnhnc \cc \pxh \kp}^2 \leq \frac{10t}{N} \cdot \norm{\pxh \kp}^2 \label{Eq:pxhncC}
    \end{equation}
  where we keep the predicate $\br{\predC}$ on the left-hand side to rule out the case where the classical query transforms the hybrid collision into a classical collision (the inequality would not hold without this predicate).

  Let $\pxh \ket{\phi} = \sum_{x,p,w,H,B,D} \alpha_{x, p, w, H,B,D} \ket{x, p, w,H,B,D} \in \alg \cap \supp{\pxh}$. By \Cref{Fact:XH}, the effect of doing a classical query on this state is given by the third item of \Cref{Lem:recordC}. Since we must not have classical collisions, we can ignore the $\ket{H_{x \la D(x)},D}$ term therein, which gives
  \begin{align*}
    \pnhnc \cc \pxh \kp =  \pnhnc \sum_{x,p,w,H,B,D} \alpha_{x, p, w, H,B,D} \ket{x,p,w} & \biggl(\frac{1}{\sqrt{N}} \ket{H_{x \la \bot},B_{\la 1},D_{x \la \bot}} \\
    & \quad - \sum_{y \in [N]} \frac{\omega^{p y}}{N} \ket{H_{x \la y},B_{\la 1},D_{x \la y}}\biggr).
  \end{align*}

  Next, using the orthogonality between basis states, the norm of the above state is at most,
  \begin{align*}
    \norm{\pnhnc \cc \pxh \kp}^2
      & \leq \sum_{\substack{x, p, w, H,B,D : \\ H(x) = \star, D(x) = \bot, \\ \predH(x,p,w,H,B,D) = \predF}} \abs[\bigg]{\sum_{z \in [N]} \alpha_{x, p, w, H,B,D_{x \la z}} \frac{1}{\sqrt{N}}}^2 \\
      & + \sum_{\substack{x, p, w, H,B,D, y : \\ y \in [N], H(x) = \star, D(x) = y, \\ \predH(x,p,w,H,B,D) = \predF}} \abs[\bigg]{\sum_{z \in [N]} \alpha_{x, p, w, H,B,D_{x \la z}} \frac{\omega^{p y}}{N}}^2.
  \end{align*}

  Hence, we can conclude in the same way as in the proof of \Cref{Eq:pnhorcQ} by using Cauchy--Schwarz inequality, which gives that $\norm{\pnhnc \cc \pxh \kp}^2 \leq \frac{2t}{N} \norm{\pxh \kp}^2$.
\end{proof}

\begin{proof}[Proof of \Cref{Eq:pxhC}]
  Let us denote $\pxh \ket{\phi} = \sum_{x,p,w,H,B,D} \alpha_{x, p, w, H,B,D} \ket{x, p, w,H,B,D} \in \alg \cap \supp{\pxh}$. By \Cref{Fact:XH}, the effect of doing a classical query on this state is given by the third item of \Cref{Lem:recordC}. Moreover, the only terms therein that can lead to a hybrid collision are those for which $D(x)$ gets replaced with a new value $y$, which gives
  \begin{align*}
      \ph \cc \pxh \kp = - \ph \sum_{x,p,w,H,B,D} \alpha_{x, p, w, H,B,D} \ket{x, p, w} \sum_{y \in [N]} \frac{\omega^{p y}}{N} \ket{H_{x \la y},B_{\la 1},D_{x \la y}}.
  \end{align*}

  Next, using the orthogonality between basis states, the norm of the above state is equal to,
  \begin{align*}
    \norm{\ph \cc \pxh \kp}^2
      & = \sum_{\substack{x, p, w, H,B,D, y : \\ y \in [N], H(x) = \star, D(x) = y}} \norm[\bigg]{\ph \sum_{z \in [N]} \alpha_{x, p, w, H,B,D_{x \la z}} \frac{\omega^{p y}}{N} \ket{x, p, w}\ket{H_{x \la y},B_{\la 1},D}}^2  \\
      & = \sum_{\substack{x, p, w, H,B,D, y : \\ y \in [N], H(x) = \star, D(x) = y, \\ \predH(x,p,w,H_{x \la y},B,D) = \predT}} \abs[\bigg]{\sum_{z \in [N]} \frac{\alpha_{x, p, w, H,B,D_{x \la z}}}{N}}^2.
  \end{align*}

  Applying the Cauchy--Schwarz inequality and rearranging the expression, we have
  \begin{align*}
    \norm{\ph \cc \pxh \kp}^2
      & \leq \sum_{\substack{x, p, w, H,B,D : \\ H(x) = \star, D(x) \neq \bot}} \abs{\alpha_{x, p, w, H,B,D}}^2 \cdot \Pr_{y \gets [N]}\bc*{\predH(x,p,w,H_{x \la y},B,D_{x \la y}) = \predT}.
  \end{align*}

  For each $(H,D)$ in the above state, $D$ contains at most $t$ entries different from $\bot$ (by definition of $\alg$). Moreover, there is exactly one hybrid collision in $(H,D)$ and this collision contains~$x$. Hence, the probability to still have a hybrid collision when $D(x)$ is replaced with a random $y \in [N]$ is at most $\Pr_{y \gets [N]}\bc*{\predH(x,p,w,H_{x \la y},B,D_{x \la y}) = \predT} \leq t/N$. We conclude that $\norm{\ph \cc \pxh \kp}^2 \leq \frac{t}{N} \norm{\pxh \kp}^2$.
\end{proof}

\begin{proof}[Proof of \Cref{Eq:phorcC}]
  The proof is almost identical to that of \Cref{lem:classical_query_progress}. The reason for which we cannot apply this lemma directly to the predicate $\predH+\predC$ is because it does not satisfy the condition stated in~\Cref{Eq:predC2}. Nevertheless, the latter equation is only needed in analyzing the projector $\Pi_2$ in the proof of~\Cref{lem:classical_query_progress}, where it is used to argue that \emph{if a basis state $\ket{x,p,w,H,B,D}$ is not in the support of $\pp$ then $\ket{x,p,w,H_{x \la D(x)},B,D}$ will not be either}. This statement is wrong for the predicate $\predP = \predH+\predC$ (indeed, if $x$ is contained in a quantum collision then $(H_{x \la D(x)},D)$ will contain a hybrid collision). However, \emph{if a basis state $\ket{x,p,w,H,B,D}$ is not in the support of $\phorc$ then $\ket{x,p,w,H_{x \la D(x)},B,D}$ will not be in the support of $\pc$}. Hence, we can carry out the same argument as in the original proof if we replace the outer projector $\pp$ with~$\pc$. This leads to $\norm{\pc \cc \pnhorc \kp}^2 \leq \frac{2 t}{N} \cdot \norm{\pnhorc \kp}^2$.
\end{proof}

\end{document}